\documentclass[11pt]{article} 

\usepackage{amsmath,amsthm,latexsym,amssymb,amsfonts,epsfig, color, comment, float}

\newtheorem{Definition}{Definition}[section]
\newtheorem{Lemma}{Lemma}[section]

\newtheorem{Proposition}{Proposition}[section]

\newcommand{\be}{\begin{equation}}
\newcommand{\ee}{\end{equation}}
\newcommand{\ba}{\begin{eqnarray}}
\newcommand{\ea}{\end{eqnarray}}

\title{{\sf Asymptotically safe canonical quantum gravity: Gaussian 
dust matter}}
\author{
{\sf R. Ferrero}$^1$\thanks{{\sf 
renata.ferrero@gravity.fau.de}}, 
{\sf T. Thiemann}$^1$\thanks{{\sf 
thomas.thiemann@gravity.fau.de}}\\
\\
{\sf $^1$ Inst. for Quantum Gravity, FAU Erlangen -- N\"urnberg,}\\
{\sf Staudtstr. 7, 91058 Erlangen, Germany}\\
}
\date{{\small\sf \today}}

\makeatletter
\@addtoreset{equation}{section}
\makeatother

\begin{document} 

\maketitle

{\sf

\begin{abstract}
In a recent series of publications we have started to investigate possible points of contact
between the canonical (CQG) and the asymptotically safe (ASQG) approach to quantum gravity, despite 
the fact that the CQG approach is exclusively for Lorentzian signature gravity while the 
ASQG approach is mostly for Euclidean signature gravity. Expectedly, 
the simplest route is via the generating functional of time ordered N-point functions 
which requires a Lorentzian version of the Wetterich equation and heat kernel methods employed in ASQG. 

In the present contribution we consider gravity coupled to Gaussian dust matter. 
This is a generally covariant Lorentzian signature system, which can be considered as a field 
theoretical implementation of the idealisation of a congruence of collision free test observers
in free fall, filling the universe. The field theory version correctly accounts for 
geometry -- matter backreaction and thus in principle serves as a dark matter model.
Moreover, the intuitive geometric interpretation selects a preferred reference frame 
that allows to disentangle gauge degrees of freedom from observables. The CQG treatment of this theory 
has already been considered in the past.

For this particular matter content it is possible to formulate the quantum field 
theory of observables as a non-linear $\sigma$ model described by a highly non-linear
conservative Hamiltonian. This allows to apply techniques from Euclidean field theory
to derive the generating functional of Schwinger N-point functions which can be treated 
with the standard Euclidean version of the heat kernel methods employed in ASQG. The 
corresponding Euclidean action is closely related to Euclidean signature gravity but 
not identical to it despite the fact that the underlying Hamiltonian is for Lorentzian
signature gravity.   
\end{abstract}

\section{Introduction}
\label{s1}

The canonical (CQG) \cite{1} and asymptotically safe (ASQG) \cite{2} approach to quantum gravity have 
received much attention in the past. To date there has been little contact between these 
programmes, mainly because CQG is exclusively for Lorentzian signature gravity while 
ASQG is mostly for Euclidean signature, see \cite{ 3aa, 3aaa, 3a} for Lorentzian work in ASQG (in particular \cite{3, 3b} for works in foliated spacetime)  and references 
therein.   

In recent work \cite{4,4a} we have started to investigate possible routes of contact between these programmes.
The fact that CQG works exclusively with the physically relevant Lorentzian signature must 
find its way in such a contact seeking enterprise. It was shown that a natural avenue is to 
formulate the generating functional of time ordered N-point functions in both frameworks 
which naturally leads to a Lorentzian version of the Wetterich equation fundamental for 
the ASQG programme and the corresponding heat kernel techniques. 

While in principle one can also attempt to construct the generating functional of Schwinger 
functions in both approaches \cite{4},
for generic matter coupling this becomes technically rather involved because when one 
integrates out the momenta one has to solve systems of partical differential equations rather 
than algebraic equations. In the present paper we consider a very particular matter content
which allows to avoid those partial differential equations \cite{5}. The classical geometry -- matter
system is described by the Einstein--Hilbert Lagrangian plus a generally covariant matter 
Lagrangian minimally coupled to geometry. Its matter content consists of scalar fields 
which give rise to a pressure free energy momentum tensor. Their covariant differentials 
define a set of four vector fields, one of which is a unit timelike geodesic tangent while 
the three others are orthogonal to it. Accordingly, the matter system can be considered 
a field theoretic modelling of a congruence of collision free massive test particles 
freely falling through the universe. The field theory formulation accounts for the fact 
that the ideal test particle does not exist and thus correctly implements geometry -- matter 
backreaction. While perhaps not entirely realistic, the matter can be considered as a dark 
matter candidate. Moreover, the natural reference frame provided by the four matter vector 
fields enables a straightforward disentangling of gauge degrees of freedom from the observables 
of the system.

Without matter, gravity carries two observable polarisations in four spacetime dimensions.
With four scalar fields serving as a material reference frame, the number of physical degrees of 
freedom is augmented to six. This can be considered as ``Higgsing the diffeomorphism 
gauge group''. One can encode those observable six degrees of freedom as two gravitational 
and and four scalar degrees of freedom (gravitational wave gauge) or one can encode them 
as six gravitational degrees of freedom (matter gauge). The analog of the latter gauge in
the electroweak interaction is the unitary gauge which reduces the four real Higgs fields 
to one and trades them for three longitudinal polarisations of massive vector bosons. 
We will choose the matter gauge in the present work due to its simplicity. The final 
picture after having removed the gauge degrees of freedom is that we obtain a non-linear
$\sigma$ model, that is, a classical field theory of dynamical 3-metrics $q$ (``symmetric matrices'') 
in four spacetime dimensions whose evolution is reigned by a conservative Hamiltonian $H=H(q,p)$
where $p$ is a symmetric matrix valued momentim conjugate to $q$. 
As one can show, the matter gauge fixes the lapse and shift of the Lorentzian signature 
four metric $g_L$ to be unity and zero 
respectively. Accordingly, implicitly this is still Lorentzian general relativity (GR) 
in the synchronous gauge dictated by a material matter reference system. 

The canonical quantisation of this model was constructed in \cite{6} using the Loop Quantum Gravity 
(LQG) choice of representation \cite{7} of the canonical commutation and adjointness relations 
among the observable fields. Furthermore, this model has been used to construct the one-loop effective action by means of its coherent state path integral representation \cite{7b}. It is therefore of considerable interest to consider the ASQG 
treatment of the system. Following the general steps laid out in \cite{1} we can construct 
the generating functional of Schwinger N-point functions starting from the canonical 
framework. This is obtained, as usual, by analytically continuing the generating functional 
of time ordered N-point functions, a step that is 
often called Wick rotation in Euclidean quantum field theory \cite{8}. The latter rely on the unitary evolution of the time 
zero fields with respect to the Hamiltonian operator mentioned above. 
That Schwinger generating functional has a path integral formulation as an integral over the phase space 
coordinatised by $q,p$ with respect to the natural Liouville measure, 
emphasising the fact that Euclidean quantum field theory is equivalent 
to classical statistical physics in four rather than three Euclidean dimensions.  
It turns out that for this model the integral over $p$ can be performed 
in closed form, a step that in general is not possible due to the appearance of partial differential 
equations as mentioned above. However, the final expression is not just the configuration 
space path integral with respect to Lebesgue measure but rather involves a Jacobian that is 
related to the the DeWitt metric which appears in the canonical formulation of GR \cite{9} (we refer the reader to \cite{9a} for a discussion on the measure in ASQG).
We get rid of that Jacobian by a field redefinition $q\mapsto Q=m(q)$ and understand the generating functional 
as generating correlation functions of that redefined metric $Q$. Classically, this is just 
a canonical transformation accompanied by $p\to P(q,p)$ and thus one bases the entire quantisation on that 
redefined field from the outset, which thus justifies this step. The absence of non-trivial measure 
factors avoids the use of additional ghost field integrals which otherwise serve to bring the 
Jacobian to the exponent. This is also why our Wetterich equation is exact without using super traces. 

After all of these preparatory steps, we end up with the Euclidean QFT formulation of a non-linear
$\sigma$ model with path integral measure simply the Lebesgue measure times the exponential 
of what one calls the Euclidean action. One can now release the 
ASQG machinery on this model, i.e. one makes use of the background field technique $Q=\bar{Q}+h$
and modifies the path integral integrand by a cutoff Gaussian in $h$ that depends on 
$\bar{Q}$ and a scale $k$. The running of the (Legendre transform of the logarithm of the) modified 
generating functional with $k$ is described by the Wetterich equation whose fixed points 
as $k\to \infty$ serve to fix the dimension free couplings of the Euclidean action of the path integral. The evaluation 
of the Euclidean action in terms of the dimensionful couplings as $k\to 0$ {\it define} the Euclidean
QFT from which one regains the canonical formulation (Hilbert space, vacuum, Hamiltonian) by 
Osterwalder-Schrader reconstruction \cite{11}. 
        
The exact Wetterich equation \cite{11a} must be truncated in practice and the computation of the truncation benefits 
from heat kernel techniques when one uses the Gaussian the background Laplacian of the background Euclidean
metric $\bar{g}$ that one obtains from $\bar{q}=m^{-1}(\bar{Q})$ by assigning unit lapse and 
zero shift. The corresponding DeWitt coefficients of the heat kernel expansion then also refer
to this restricted class of Euclidean signature metrics. Note that within this restricted 
class of metrics the shift between signatures is simply by analytic continuation of the lapse between 
the real and imaginary unit which is rigorously possible here because both lapse and shift are no longer integrated over.
See \cite{12} for general considerations of the lapse analytic continuation in ASQG.

In this paper we restrict to the Einstein--Hilbert truncation of the Wetterich equation as a first 
step \cite{11b}. We do not rely on cutoff Gaussians or cutoff kernels of the type usually employed 
in ASQG because these rely on the unproved assumption that these have a pre-image under 
the Laplace transform. In \cite{4} we have shown that the question about the existence of the Laplace pre-image
is non-trivial and can be answered in the negative for some of the suggested cutoff functions.
Instead we employ cutoff functions defined as the Laplace transform 
of a natural and concrete choice of pre-image cutoff functions. These pre-images are
Schwartz functions with respect to both heat kernel time and its inverse which ensures existence 
of otherwise singular heat kernel time integrals. The evaluation of the corresponding 
flow equations benefits from tools associated with the Barnes type integrals \cite{13} (we make use of the automatisation implemented in \cite{barnes}).      
We solve the flow and analyse the fixed point structure within the current truncation and compare 
with the literature.\\
\\
This article is organised as follows:\\
      
In section \ref{s2} we define the classical Gaussian dust model and sketch a few steps 
of the canonical classical and quantum analysis.

In section \ref{s3} we derive the formal path integral of the generating functional
of connected Schwinger functions. A new element of our treatment is to use a non-standard 
density weight to avoid otherwise non-trivial measure corrections as outlined above. It turns out 
that the Euclidean action is essentially the Einstein--Hilbert Lagrangian for 
Euclidean signature in synchronous gauge although the Hamiltonian comes from the Lorentzian signature 
Lagrangian. We explain why this is no contradiction and how Wick rotation has to be 
understood, in particular why one does not end up with a path integral for 
complex GR.  

Section \ref{s4} is devoted to a thorough discussion of cutoff functions and heat kernels
for the present theory which has a reduced symmetry group, namely active rather than 
passive diffeomorphisms that preserve the synchronous gauge. The cutoff function 
must be invariant with respect to that reduced symmetry group only which enhances 
the freedom in this choice. We pick a cutoff which is induced as much as possible  
from the standard choice of cutoff combined with a natural projection operator. 

In section \ref{s5} we compute the Einstein--Hilbert truncation of the corresponding 
Wetterich equation and analyse the flow equations, fixed points, critical exponents 
and $k=0$ limits. This involves, in addition to the usual heat kernel expansion, an 
additional expansion in the polynomial degree of non-minimal operators which involve 
multiple commutators between the standard Laplacian and the afore mentioned projection
operator. In this paper we focus as a first step on the lowest order of that additional
expansion in terms of the non-minimal operators. 

In section \ref{s6} we summarise, conclude and give an outlook.

In appendix \ref{sa} we sketch the afore mentioned expansion of the projected heat 
kernel with respect to the polynomial degree in non-minimal operators. 
  
In appendix \ref{sb} for the benefit of the reader  we include some background information on the Barnes integral technique 
that we use for our concrete choice of cutoff.           

\section{Gaussian dust}
\label{s2}

In the first subsection we briefly review the classical starting point of the theory
under consideration. In the second we perform a canonical transformation on the 
classical reduced phase space which is motivated in the third subsection.

\subsection{Review of the classical canonical treatment}
\label{s2.1}

We follow closely \cite{5,6} but generalise the analysis to arbitrary spacetime dimension.\\
\\
The generally covariant Gaussian dust Lagrangian density reads
\be \label{2.1}
l_{GD}=-|\det(g)|^{1/2}\;[\frac{\rho}{2}\;(g^{\mu\nu} \; 
T_{,\mu}\; T_{,\nu}+1)+g^{\mu\nu}\;T_{,\mu}\; (W_j\; S^j_{,\nu})]\;.
\ee
Here $\mu=0,1,...,D$ and $j=1,..,D$ where $D+1$ is the spacetime 
dimension. Thus it depends on $2(D+1)$ scalar fields $(T,S^j),\;(\rho,W_j)$. The Euler-Lagrange 
equations for $\rho,W_j$ respectively yield $g(U,U)=-1$ and $g(U,V_j)=0$ where 
$U^\mu=g^{\mu\nu} T_{,\nu}, V^\mu_j=g^{\mu\nu} S^j_{,\nu}$. This already implies that 
$\nabla_U U=0$, i.e., that $U$ is a unit timelike geodesic tangent.
The Euler-Lagrange equations for 
$T,S^j$ respectively yield the conservation 
equations $\nabla_\mu(\rho U^\mu+W^j V_j^\mu)=\nabla_\mu (W_j \;U^\mu)=0$ which imply 
that the energy momentum tensor $T^{\mu\nu}=\rho\;U^\mu U^\nu+2\;W^j V_j^{(\mu}\; U^{\nu)}$ 
is conserved. The pressure $p=\frac{1}{D}(g^{\mu\nu}+U^\mu U^\nu)\; T_{\mu\nu}=0$ vanishes
which motivates the attribute ``dust''.

The canonical analysis of the system proceeds
via a $D+1$ split of spacetime $M\cong \mathbb{R}\times \sigma$ where
$\sigma$ is a $D$-manifold and a corresponding ADM parametrisation of $g$ \cite{9} in terms 
of the pull-back metric $q$ on $\sigma$ and lapse $N$ and shift $N^a$ functions. 
Then (\ref{2.1}) becomes
\be \label{2.2}
l_{GD}=-N\;|\det(q)|^{1/2}\;[\frac{\rho}{2}\;(-[\nabla_n T]^2+q^{ab} T_{,a} T_{,b}+1)
+W_j(-[\nabla_n T]\;[\nabla_n S^j]+ q^{ab} T_{,a} S^j_{,b})]\;,
\ee
where $n=\frac{1}{N}[\partial_t-N^a \partial_a],\; a=1,..,D$ is the timelike unit normal
to the $t=$ const. surfaces. Computing the momenta $\partial l/\partial (\partial_t (.))$
conjugate to the eight scalar fields yields primary constraints $Z=Z^j=0$ 
where $Z, Z^j$ are the momenta conjugate to $\rho, W_j$ and $\zeta_I=P_I-\frac{W_I}{W_D} P_D,\; I=1,..,D-1$ 
where where $P, P_j$ are the momenta conjugate to $T, S^j$. The Legendre transform 
yields the Hamiltonian density (abbreviating $w=\sqrt{\det(q)}$)
\ba \label{2.3}
h_{GD} &=& v\; Z+v_j Z^j+[u^I-N w\frac{W_D q^{ab} T_{,a} S^I_{,b}}{P_D}] \zeta_I+
N^a[P\;T_{,a}+P_j S^j_{,a}]
\nonumber\\
&& +N[w^{-1}\{P\frac{P_D}{W_D}-\frac{\rho}{2}(\frac{P_D}{W_D})^2\}
+w\{\frac{\rho}{2}[1+q^{ab} T_{,a} T_{,b}]+\frac{W_D}{P_D}\;q^{ab}\; T_{,a} P_j S^j_{,b}\}]\;,
\ea
where $v,v_j,u^I$ are the velocities that one cannot solve for.
 
This has to be supplemented by the geometry contribution which 
we take as Einstein--Hilbert Lagrangian with 
cosmological constant $l=\kappa^{-1}|\det(g)|^{1/2}[R(g)-2\Lambda]$ with 
Ricci scalar $R(g)$ of the $D+1$ metric $g$. The canonical analysis yields additional 
primary constraints $p=p_a=0$ where $p,p_a$ are the momenta conjugate to $N,N^a$ while 
the momenta conjugate to $q_{ab}$ are denoted as $p^{ab}$. This yields the well known 
result \cite{9}
\be \label{2.4} 
\kappa h=u P+u^a P_a+N^a (-2 D_b p^b\;_a)+N[w^{-1}\{(p_{ab}\; p^{ab})^2-\frac{1}{D-1}
(p^a\;_a)^2\}-w\; (R(q)-2\Lambda)]\;,
\ee
where $u,u^a$ are again non-solvable velocities, $\kappa$ is proportional 
to Newton's constant and $\Lambda$ the cosmological 
constant. Here
spatial indices are moved with $q_{ab},\; q^{ab}, q_{ac} q^{cb}=\delta_a^b$, 
$D_a$ is the torsion free covariant differential compatible with $q$ and $R$ its 
Ricci scalar. The result (\ref{2.4}) holds for Lorentzian signature. Remarkably,
for Euclidean signature one just has to invert the sign in front of the term quadratic in 
$p_{ab}$.   

The total Hamiltonian density $h_T=h_{GD}+h$ generates equations of motion via 
Poisson brackets. To ensure that the primary constraints 
$p_a,\;p,\;Z_j,Z,\zeta_I$
are preserved in time 
one uses Dirac's algorithm \cite{14}. This leads to secondary constraints
\ba \label{2.5}
k_D &=& -\frac{1}{2}(\frac{P_D}{W_D})^2+w\frac{1}{2}[1+q^{ab} T_{,a} T_{,b}]
\nonumber\\
k &=& N[w^{-1}\{-P\frac{P_D}{W_D^2}+\rho\frac{P_D^2}{W_D^3}\}
+w\frac{1}{P_D}\;q^{ab}\; T_{,a} P_j S^j_{,b}]
\nonumber\\
c^T_a &=& c_a^{GD}+c_a,\;
c_a^{GD}=P\;T_{,a}+P_j S^j_{,a},\;\;\kappa c_a=-2 D_b p^b\;_a;\;\;
c^T=c^{GD}+c
\ea
where $c^{GD}, c$ are the coefficients of $N$ in (\ref{2.3}) and (\ref{2.4}) respectively.
The constraints $k,k_D$ result from stabilising $Z,Z_D$ respectively while $c^T,c^T_a$ result from 
stabilising $p,p_a$ respectively. Stabilising $Z_I,\zeta_I$ respectively fixes
$v_I,u^I$ respectively. Stabilising the secondary constraints produces no new 
constraints but fixes $v,v_D$. The total list of constraints is now subdivided into those 
of first class $p,p_a,c^T,c_a^T$ and second class pairs $(Z,k),\;(Z^D,k_D),\;(Z^I,\zeta_I)$.
We are asked to compute the corresponding Dirac bracket and to solve the second class constraints.
We solve $k_D,k=0$ respectively for $W_D,\rho$ in the form  
\be \label{2.6}
\frac{P_D}{W_D}=w\sqrt{1+q^{ab} T_{,a} T_{,b}},\;\;
\rho=\frac{W_D}{P_D}[P-w^2\frac{q^{ab}\; T_{,a} P_j S^j_{,b}}{(P_D/W_D)^2}]
\ee 
and can insert this into the remaining constraints as we use the Dirac bracket 
instead of the Poisson brackets. This yields
\be \label{2.7}
c^T=c+\frac{P-q^{ab} T_{,a} c_b}{\sqrt{1+q^{cd} T_{,c} T_{,d}}}\;,
\ee
where $c^T_a=0$ was used. The constraints $\zeta_I$ can be solved for $W_I$. In this way,
all variables $(\rho,Z),\; (W_j,Z^j)$ have completely disappeared.

After the second class constraints have been solved we can focus on the left over variables 
and constraints $p,p_a,c^T, c^T_a$ which just involve the canonical pairs          
$(N,p)\;(N^a,p_a),\;(q_{ab},p^{ab})$. As the difference between Dirac and Poisson bracket 
involves terms that contain at least one Poisson bracket with $Z,Z^j$, the Dirac bracket 
coincides with the Poisson bracket on functions of our remaining canonical pairs.
The Hamiltonian density is now
\be \label{2.8}
h_T=\kappa^{-1}[u\; P+u^a\; P_a]+N\; c^T+N^a\;c_a^T]\;.
\ee
We completely reduce the phase space by imposing gauge conditions on the variables 
$T,S^j$ namely
\be \label{2.9}
T=T_\ast=t,\;S^j=S^j_\ast=\delta^j_a\; x^a\;.
\ee
The stability of these gauge conditions under the Hamiltonian flow generated by
$\frac{d}{dt}(.)=\partial_t+\{\int_\sigma \; d^Dx h_T(x),.\}$
now fixes lapse and shift
\be \label{2.10}
N=N_\ast=1,\; N^a=N^a_\ast=0
\ee
and we solve $c^T,c_a^T=0$ for $P^\ast=-c,\;P_j=P^\ast_j=-\delta^a_j\; c_a$.
The stability of the fixed lapse and shift fixes $u=u_\ast=0,\;u^a=u^a_\ast=0$
and we have trivially $Z=Z_\ast=0,\;Z^j=Z^j_\ast=0$.
The reduced phase space is coordinatised by the true degrees of freedom or observables 
$q_{ab},p^{ab}$ and the reduced Hamiltonian on functions $F$ of those is given by
\be \label{2.11}
\{H,F\}:=\{\int\; d^Dx\; h_T(x),F\}_{u=u_\ast,N=N^\ast,Z=Z_\ast,P=P_\ast,T=T_\ast,S=S_\ast}
=\{\int\; d^Dx \;c(x), F\}\;.
\ee
In other words, we end up with a conservative Hamiltonian system defined by a non-linear
$\sigma$ model of matrices $q$ with conjugate momentum $p$ and Hamiltonian 
\be \label{2.12}
H=\kappa^{-1}\;\int_\sigma\; d^Dx\;
[w^{-1}\{(p_{ab}\; p^{ab})^2-\frac{1}{D-1}(p^a\;_a)^2\}-w\; (R(q)
-2\Lambda)]\;.
\ee

Note that (\ref{2.12}) is no longer constrained to vanish, it is not the generator 
of temporal diffeomorphism gauge transformations but rather of physically observable 
time translations. It even has an infinite number of conserved charges: For every vector 
field $w^a$ on $\sigma$, the functional $c[u]:=\int_\sigma\; d^Dx\; u^a\; c_a$ is a 
constant of motion. Note that the $c[u]$ are also no longer constrained to vanish,
they do not generate spatial diffeomorphism gauge transformations but 
rather physically observable active spatial diffeomorphisms. The easiest way to check this
is to recall the hypersurface deformation algebroid relations \cite{9} 
\be \label{2.12a}
\{c[u],c[v]\}=-c[[u,v]],\;  
\{c[u],c[f]\}=-c[u[f]],\;
\{c[f],c[g]\}=-c[q^{-1}(f\;dg-g\;df)],v]]\;,
\ee
where $[u,v],\; u[f]$ are respectly the Lie derivatives of the vector field $v$ and scalar
$f$ respectively with respect to $u$ and to note that $H=c[f=1]$. Here 
with $c[f]:=\int_\sigma\; d^Dx\; f\; c$.

\subsection{Canonical transformation}
\label{s2.2}

For reasons that will become transparent in the next subsection, we consider passing to new canonical 
configuration coordinates 
\be \label{2.13}
Q_{ab}=[\det(q)]^r\; q_{ab} \;\; \Rightarrow\;\; \det(Q)=\det(q)^{1+rD},\;
q_{ab}=[\det(Q)]^{-\frac{r}{1+rD}}\; Q_{ab}  
\ee
for some $r\in \mathbb{R}$. The conformally rescaled metric $Q$ is a twice covariant 
symmetric tensor field of spatial density weight $2r$ with respect to spatial 
diffeomorphisms. We can complete this to a canonical transformation by passing
to new momenta
\be \label{2.14}
P^{ab}=[\det(q)]^{-r}\;[p^{ab}-\frac{r}{1+rD}\; q_{cd}\; p^{cd} \; q^{ab}]\;\;
\Rightarrow\;\; P^{ab} Q_{ab}=\frac{1}{1+rD} p^{ab}\; q_{ab},\;
p^{ab}=[\det(Q)]^{\frac{r}{1+rD}}[P^{ab}+r\; Q_{cd} P^{cd} Q^{ab}] 
\ee
where $q^{ac} q_{cb}=Q^{ac} Q_{cb}=\delta^a_b$. That is, the non-vanishing Poisson brackets are
\be \label{2.15}
\{P^{ab}(x), Q_{cd}(y)\}=\{p^{ab}(x), q_{cd}(y)\}=\kappa\delta^{(D)}(x,y)\; \delta^{(a}_c \;\delta^{b)}_d\;.
\ee
The Hamiltonian reads in terms of $Q,P$
\be \label{2.16}
H=\int\; d^Dx\;
\{[\det(Q)]^{-\frac{1}{2(1+rD)}}[Q_{ac}\;Q_{bd}-\frac{1+2r+r^2 D}{D-1} Q_{ab}\;Q_{cd}]\; P^{ab}\; P^{cd}
-[\det(Q)]^{\frac{1}{2(1+rD)}}\; [R(\det(Q)^{-\frac{r}{1+rD}}\; Q)-2\Lambda]\} 
\ee
The transformation of the Ricci scalar under conformal transformations 
$q_{ab}=\Omega^2 Q_{ab},\; \Omega^2=\det(Q)^{-\frac{r}{1+rD}}$
is given by the formula \cite{9} 
\be \label{2.17}
R(q)=\Omega^{-2}\;\{R(Q)
-2(D-1) Q^{ab} D^Q_a D^Q_b \ln(\Omega)
-(D-2)(D-1) Q^{ab} [D^Q_a \ln(\Omega)] [D^Q_b \ln(\Omega)]\}\;,
\ee
where $D_a^Q$ is the covariant differential of $Q$ as if it had density weight zero,
$R(Q)$ is the Ricci scalar of $Q$ as if it had density weight zero and $\Omega$ is treated as a 
scalar of density weight zero, that is
\ba \label{2.18}
D_a^Q \ln(\Omega)=\partial_a \ln(\Omega)
&=& -\frac{r}{2(1+rD)} Q^{bc} Q_{bc,a}
=-\frac{r}{1+rD} \Gamma^b_{ba}(Q),\;\;
\nonumber\\
Q^{ab} D_a^Q D_b^Q\ln(\Omega) &=& -\frac{r}{1+rD}\;Q^{ab}
[\partial_a \Gamma^c_{cb}(Q)-\Gamma^c_{ab}(Q)\;\Gamma^d_{dc}(Q)]\;.
\ea
Note that since both $Q$ and $\Omega$ have non-trivial density weight for $r\not=0$, 
neither $R(Q)$ nor the additional terms in (\ref{2.17}) are scalars (density weight zero)
but they combine altogether to a scalar.  

\subsection{Measure Jacobians}
\label{s2.3} 

The motivation for introducing $Q,P$ is as follows: When we construct the path integral for
our theory, in a first step one arrives at a path integral over phase space with 
Liouville measure in $Q,P$ times the exponential of the Hamiltonian. Since the Hamiltonian
is quadratic in $P$, one can integrate out the momenta to arrive at the usual configuration 
path integral but there is a Jacobian left over coming from doing the Gaussian integral.
We will determine $r$ such that this Jacobian is independent of $Q$ so that the configuration
path integral measure is simply the Lebesgue measure in $Q$. We are thus concerned with the 
term in (\ref{2.16}) which is quadratic in $P$. The discussion is simplified by introducing
a (density valued) D-Bein $e^a_i,\; i=1,..,D$ with inverse $e_a^i$ satisfying 
$Q_{ab}=\delta_{ij} e^i_a e^j_b$ and the objects
\be \label{2.19}
\hat{P}^{ij}:=e^i_a\; e^j_b\; P^{ab} [\det(Q)]^{-\frac{1}{4(1+rD)}}\;\;\
\Rightarrow\;\; P^{ab}=e^a_i e^b_j \hat{P}^{ij} [\det(Q)]^{\frac{1}{4(1+rD)}} \;.  
\ee
In terms of the new integration variables $P^{ij}$ the quadratic term is simply
\be \label{2.20}
[\delta_{i(k} \delta_{l)j}-u \;\delta_{ij} \delta_{kl}]\; \hat{P}^{ij} \hat{P}^{kl},\;
u=\frac{1+2r+r^2 D}{D-1}\;.
\ee
We pick some gauge of the internal rotation freedom that is left over in determining 
$e_a^i$ given $Q_{ab}$ and think of $e_a^i=e_a^i(Q)$ as determined entirely by $Q$.
Then the Liouville measure transforms as 
\be \label{2.21}
d\mu_L(Q,P)=[d^k Q]\; [d^k P]=[d^k Q]\; [d^k \hat{P}]\; |\det([\partial P]/[\partial \hat{P}])|\;,
\ee
where $k=D(D+1)/2$ is the dimension of the space of symmetric tensors. The Jacobian is given
explicitly by 
\be \label{2.22}
|\det([\partial P]/[\partial \hat{P}])|=[\det(Q)]^{\frac{k}{4(1+rD)}}
|\det(E)|,\; E^{ab}_{ij}=e^a_{(i}\; e^b_{j)}\;.
\ee
The matrix $E$ is a $k\times k$ matrix. To compute its determinant we introduce a 
lexicographic ordering of symmetric index pairs $A=(a,b),\; B=(c,d)$ with 
$a\le b,\;c\le d$: We define 
$A<B$ if either $a<c$ and no condition on $b,d$ or $a=b$ and $b<d$. We define $A=B$
if $a=c$ and $b=d$. In this way we have a one to one correspondence between symmetric 
index pairs $(a,b)$ and values $A=1,..,k$. We define the same ordering for 
symmetric index pairs $I=(i,j),\; J=(k,l)$ with $i\le j,\; k\le l$. Then we have 
the following result.
\begin{Lemma} \label{la2.1} ~\\
Suppose that $e^a_i$ is an upper triangular $D\times D$ matrix, i.e. $e^a_i=0$ for $i<a$
where the $a$ type index labels rows and the $i$ type index labels columns. 
Then $E^A_I$ is upper triangular $k\times k$ matrix
with respect to the lexicographic ordering defined above, that is, 
$E^A_I=0$ for $I<A$.
\end{Lemma}
\begin{proof}: \\
Suppose that $I=(i,j)<A=(a,b)$ with $i\le j,\;a\le b$ and consider 
$E^A_I=\frac{1}{2}(e^a_i e^b_j+e^a_j e^b_i)$. Thus either I. $i<a$ and no 
condition on $j,b$ or II. $i=a$ and $j<b$. In case I. 
we have $e^a_i=0$ because $i<a$ and $e^b_i=0$ because $i<a\le b$, hence $E^A_I=0$.
In case II. we have $e^b_j=0$ because $j<b$ and $e^b_i=0$ because $i\le j<b$, hence 
$E^A_I=0$.
\end{proof}
\begin{Lemma} \label{la2.2} ~\\
If $e^a_i$ is upper triangular then $\det(E)=(\frac{1}{2})^{k-D}\;[\det(e)]^{D+1}$.
\end{Lemma}
\begin{proof}: \\
By lemma \ref{la2.1}, $E$ is upper triangular and the determinant of an upper triangular 
matrix is the product of its diagonal entries $\det(E)=\prod_{A=1} E^A,\;
E^A:=[E^A_I]_{I=A}=E^A_I \delta^I_A$. Using again 
$E^A_I=\frac{1}{2}(e^a_i e^b_j+e^a_j e^b_i)$ with $A=I$ i.e. $a=i, \; b=j$ we find 
for $a=b$ that $E^A=(e^a)^2$ with $e^a=[e^a_i]_{i=a}=e^a_i \delta^i_a$ while for $a<b$ we have 
$E^A=\frac{1}{2} e^a e^b$ because $e^b_i=0$ for $i=a<b$. There are $k-D$ diagonal 
entries of $E$ with $a<b$ hence $\det(E)=(\frac{1}{2})^{k-D}\;\prod_{1\le a\le b \le d}
e^a e^b$. Given $1\le c\le d$, we ask how many times the factor $e^c$ occurs in this 
product of $2k=D(D+1)$ factors of the $e^d$. It occurs twice in $e^a e^b$ for $a=b=c$,
it occurs once in $e^a e^b$ for $c=a<b$ i.e. $b=c+1,.., D$ which are $D-c$ possibilities 
and it occurs once in $e^a e^b$ for $a<b=c$ i.e. $a=1,..,c-1$ which are $c-1$ possibilities.
Thus each $e^c$ occurs $2+D-1=D+1$ times in the product. It follows 
$\prod_{a\le b} e^a e^b=(\prod_c e^c)^{D+1}=[\det(e)]^{D+1}$ where we used again that 
$e$ is upper triangular.
\end{proof}
\begin{Proposition} \label{prop2.1} ~\\
The Jacobian (\ref{2.22}) is constant on the phase space iff 
\be \label{2.23}
r=\frac{D-4}{4\;D}
\ee
\end{Proposition}
\begin{proof}:\\
Given a positive definite metric $Q_{ab}$ we find a unique co-D-Bein such that 
$e^a_i(Q)$ is upper triangular with positive diagonal entries (see e.g. \cite{15}).
Combining (\ref{2.22}) with lemmas \ref{la2.1}, \ref{la2.2} the Jacobian is  
\be \label{2.24}
|\det([\partial P]/[\partial \hat{P}])|=2^{D-k}\;
[\det(Q)]^{\frac{k}{4(1+rD)}-\frac{D+1}{2}}
\ee
where we used that $\det(\{Q^{ab}\})=|\det(\{e^a_i\})|^2=\det(Q)^{-1}$. This becomes 
independent of $Q$ when $D=2\frac{k}{D+1}=4(1+rD)$ i.e. $r=\frac{D-4}{4\;D}$.
\end{proof}
Remarkably for $D=4$ nothing needs to be done while for the physically interesting case 
$D=3$ we have $r=-\frac{1}{12}$, i.e. $Q$ has density weight $-\frac{1}{6}$. 

With the tools provided one now has to integrate over $\hat{P}\in \mathbb{R}^k$ the 
function $\exp(-M_{IJ} \hat{P}^I \hat{P}^J)$ using the Lebesgue measure
$d^k \hat{P}$ where $M_{IJ}=\delta_{IJ}-u\;\delta_I\delta_J$ where 
$\delta_{I}=\delta_{ij}$ for $I=(i,j)$ and $\delta_{IJ}=\delta_{k(i}\;\delta_{j)k}$,
see (\ref{2.20}). It is important to note that in the Einstein summation convention 
of say $u_I v^I$ we {\it do not} mean the lexicographic summation
$\sum_{I=1}^k u_I v^I=\sum_{i=1}^D u_{ii} v^{ii}+\sum_{1\le i<j\le D} u_{ij} v^{ij}$
but rather the original tensor summation
$\sum_{i,j=1}^D u_{ij} v^{ij}=\sum_{i=1}^D u_{ii} v^{ii}+2\sum_{1\le i<j\le D} u_{ij} v^{ij}$.
The lexicographic ordering was introduced just in order to keep track of the number
of independent summation variables which has an influence on the form of the Jacobian. 
It is easy to diagonalise the matrix $M$ in the space of symmetric matrices 
$T^I=T^{ij}$: The vector $\delta^I$ is an eigenvector with eigenvalue 
$1-u D=\frac{D-1-D(1+2r+r^2 D)}{D-1}$
which is easily checked to be negative for all values $D\ge 1$. In the $k-1$ dimensional 
space of trace-free vectors $\delta_I T^I=\delta_{ij} T^{ij}=0$ pick an orthonormal 
basis $T^I_\alpha,\;\alpha=1,..,k$ with respect to the metric $\delta_{IJ}$. Then 
$T^I_\alpha$ has eigenvalue unity with respect to $M$. It follows that the 
DeWitt metric $M$ has signature $(1,k-1)$. Thus after switching to adapted integration 
variables at the price of a numerical constant Jacobian we are confronted with 
fractions of integrals of the form $Z[F]/Z[1]$ where 
\be\label{2.25}
Z[F]=\int\;d^k y\;\int\; d^k x\; \exp(-\eta_{IJ} x^I x^J+i x^I y_I) \; F(y)\;,
\ee 
where $\eta$ is the Minkowski metric in $k$ dimensions. We regularise this divergent integral 
by replacing $(x^0)^2$ by $-a(x^0)^2$ for $a>0$ resulting in $Z_a[F]$ and define 
$Z[F]$ as the analytic continuation of $Z_a[F]$ to $a=-1$. One obtains
\be \label{2.26}
\frac{Z_a[F]}{Z_a[1]}=\frac{z_a[F]}{z_a[1]},\;
z_a[F]=\int\;d^k y\;
\exp(-\frac{1}{4}[y_0^2/a+\sum_{A=1}^{k-1} y_k^2])\; F(y)
\to \int\;d^k y\;
\exp(-\frac{1}{4}\eta^{IJ} y_I y_J)\; F(y)\;.
\ee

More precisely, the path integral involves the exponential of 
\be \label{2.27}
-\int\;dt\;\int\; d^Dx\; \{[\det(Q)]^{-\frac{1}{2(1+rD)}}[Q_{ac} Q_{bd}-u Q_{ab} Q_{cd}]
P^{ab} P^{cd}+i\dot{Q}_{ab} P^{ab}-[\det(Q)]^{\frac{1}{2(1+rD)}}\;[R([\det(Q)]^{-\frac{r}{1+rD}} Q)-2\Lambda]\}
\ee
Set $\hat{Y}_{ij}=\dot{Q}_{ab} e^a_i e^b_j [\det(Q)]^{\frac{1}{4(1+rD)}}$ then the
$P$ dependent part of the integrand of (\ref{2.27}) can be written
\be \label{2.28}
M_{IJ} \hat{P}^I \hat{P}^J+iY_I \hat{P}^I
=M_{IJ}[\hat{P}^I+\frac{i}{2}\;M^{IK} Y_K]\;[\hat{P}^J+\frac{i}{2}\;M^{JL} Y_L]\;
+\frac{1}{4} M^{IJ} Y_I Y_J\;,
\ee
where $M^{IK} M_{KJ}=\delta^I_J$. The Ansatz $M^{IJ}=\delta^{IJ}-u' \delta^I \delta^J$ 
gives $M^{IK} M_{KJ}=[\delta^I_J-u\delta^I\delta_J]-u'\delta^I[1-uD]\delta_J$ hence 
$u'=\frac{u}{uD-1}$. It follows that after integrating over $\hat{P}$ under the assumption 
$u<0$ and then continuing to the actual positive value of $u$ (\ref{2.27}) becomes
\be \label{2.29}
-\int\;dt\;\int\; d^Dx\; \{
\frac{1}{4}[\det(Q)]^{\frac{1}{2(1+rD)}}[Q^{ac}\;Q^{bd}-u'\;Q^{ab}\;Q^{cd}]\dot{Q}_{ab}\dot{Q}_{cd}
-[\det(Q)]^{\frac{1}{2(1+rD)}}\;[R([\det(Q)]^{-\frac{r}{1+rD}} Q)-2\Lambda]\}\;.
\ee
We check that the term quadratic in $\dot{Q}$ inside the curly bracket of (\ref{2.29})
reduces to the extrinsic curvature term when transforming back 
from $Q$ to $q$. Indeed this gives using $\dot{Q}=[\det(q)]^r\;[\dot{q}+r\;v\;q],\;v=q^{ab}\dot{q}_{ab}$
\ba \label{2.30}
&& \frac{1}{4}[\det(q)]^{1/2}\;[q^{ac} \; q^{bd}-u'\;q^{ab}\; q^{cd}]\;
[\dot{q}_{ab}+r\;q_{ab}\;v]\;
[\dot{q}_{cd}+r\;q_{cd}\;v]\;
\nonumber\\
&=&
\frac{1}{4}[\det(q)]^{1/2}\;[q^{ac} \; q^{bd}
\;\dot{q}_{ab}\;\dot{q}_{cd}+[2r+D\;r^2-u'(1+rD)^2]\;v^2]
\nonumber\\
&=& \frac{1}{4}[\det(q)]^{1/2}\;[q^{ac} \; q^{bd}-q^{ab} \; q^{cd}]\dot{q}_{ab}\;\dot{q}_{cd}\;,
\ea
where we used $2r+r^2 D=-1+u(D-1), u'=u/(uD-1)$. However, at unit lapse and zero shift 
the extrinsic curvature is
\be \label{2.31}
k_{ab}=\frac{1}{2N}\;[\dot{q}_{ab}-[{\cal L}_{\vec{N}} q]_{ab}]=\frac{\dot{q}_{ab}}{2}
\ee
which means that (\ref{2.29}) equals in terms of $q,k$
\be \label{2.33bis}
-\int\;dt\;\int\; d^Dx\; [\det(q)]^{1/2}\{[q^{ac} \; q^{bd}-q^{ab} \; q^{cd}]\;k_{ab}\;k_{cd}
-[R(q)-2\Lambda]\}\;.
\ee 
Now recall the Codacci relation (e.g. \cite{9})
\be \label{2.33}
R_\epsilon(g)=R(q)-\epsilon\;[k^{ab} k_{ab}-(k^a\;_a)^2]+2\epsilon\nabla_\mu\;[
n^\nu (\nabla_\nu n^\mu)-n^\mu (\nabla_\nu n^\nu)]\;,
\ee
where $\epsilon=\pm 1$ for Euclidean/Lorentzian signature GR respectively and $\nabla$ 
is the covariant differential of $g$. The 
last term times $|\det(g)|^{1/2}$ is a total divergence. Using $n^\mu=\delta^\mu_t$ 
and $g_{tt}=\epsilon,\; g_{ta}=0$ for unity lapse and zero shift one finds 
that it equals $\frac{\partial_t^2\sqrt{\det(q)}}{\sqrt{\det(q)}}$. Therefore 
(\ref{2.33bis}) equals (using that $\det(g)=\epsilon \det(q)$)
\be \label{2.35}
+\int_{\mathbb{R}\times\sigma}\;d^{D+1} x \; \sqrt{\det(g)} R_{\epsilon=+1}(g)
-2\lim_{T\to \infty}\; [\partial_t[\int_\sigma\; d^Dx\; \sqrt{\det(q)}]\}_{t=-T}^{t=T}\;.  
\ee
The Gibbons-Hawking boundary term \cite{16} is equal to the difference between the 
distant future and past of the time derivative of the total volume of the universe
(due to $\dot{\sqrt{\det(q)}}=\sqrt{\det(q)} q^{ab}\;k_{ab}$ it can also be written 
in terms of the trace of the extrinsic curvature). 
The volume term is the Einstein--Hilbert action of {\it Euclidean signature} GR although
the Hamiltonian corresponds to {\it Lorentzian signature} GR. This happens because
of the analytic continuation of time that we performed. Note, however, that this 
does not mean that the time argument of the metric becomes imaginary as it is 
often criticised. The path integral is strictly over real valued spacetime and spacetime metrics 
with Euclidean signature in {\it synchronous gauge} (unit lapse, zero shift) and thus 
reduces to a path integral over just the dynamical spatial metric.    
 
In what follows we will not ignore the boundary term by restricting the path of $D-$metrics 
to those which have the same total volume derivative in the distant past and future but rather 
keep the action in the form (\ref{2.29}). 
         
\section{Quantisation}
\label{s3}

In this section we quantise the classical theory of the previous section based on the 
polarisation $(Q,P)$ and the Hamiltonian $H$. We use canonical 
quantisation and from that derive a path integral formulation. The inverse path from 
the path integral to canonical quantisation is known as Osterwalder-Schrader reconstruction 
on which we briefly touch at the end of the section. We will follow closely \cite{4} and will 
be brief in our presentation which is included only for reasons of self-containedness.\\
\\
The starting point is the Weyl algebra $\mathfrak{A}$ of the {\it time zero fields}
$Q,P$. It is generated by the Weyl elements
\be \label{3.1}
w[f]=e^{i\;Q(f)},\;w[g]=e^{i \;P(g)},\;
Q(f)=\int_\sigma\;d^D x\; f^{ab}\; Q_{ab},\; 
P(g)=\int_\sigma\;d^D x\; g_{ab}\; P_{ab},\; 
\ee
which are subject to the canonical commutation relations (CCR) 
\be \label{3.2}
w[g]\;w[f]\;w[-g]=e^{-ig(f)}\;w[f]
\ee
(all other commutators are trivial and we have set $\hbar=1$) and adjointness relations (AR)
\be \label{3.3}
w[f]^\ast=w[-f],\;w[g]^\ast=w[-g]\;.
\ee
We then consider the representation theory of $\mathfrak{A}$. We consider cyclic representations
$(\rho,{\cal H})$ of the elements $a\in \mathfrak{A}$ by operators $\rho(a)$ on a Hilbert space
$\cal H$ which are such that there exists a unit vector $\Omega\in {\cal H}$ such that 
${\cal D}:=\rho(\mathfrak{A})\Omega$
is dense. As is well known, such representations are equivalently defined by positive, linear
normalised functionals $\omega$ on $\mathfrak{A}$ where the correspondence is given by 
$\omega(a)=<\Omega,\;\rho(a)\Omega>_{{\cal H}}$. In QFT there is no uniqueness theorem 
on the choice of $\omega$ and thus to select a suitable $\omega$ we use the physical input 
that $\omega$ allows for a quantisation of the Hamiltonian $H$ as a self-adjoint operator 
$\rho(H)$ densely defined on $\cal D$. This step is very hard and for interacting theories 
not yet under rigorous control. However, assuming this to be the case we have at our disposal
the unitary operators $U(t)=e^{-it \rho(H)}$ which define the Heisenberg evolution 
$\rho(a)(t)=U(t)\rho(a)U(-t)$. We also assume that $\rho(H)$ has a ground state $\Omega_H$,
i.e. an eigenvector of lowest eigenvalue (by shifting $H$ by a constant we can assume 
that eigenvalue to be zero). 

We are then interested in the time ordered correlators
\be \label{3.4}
<\Omega_H,\; \rho(a_N)(t_N)\;..\;\rho(a_1)(t_1)\;\Omega_H>\;,
\ee
with $t_{l+1}>t_l,\;l=1,..,N-1$ which occur e.g. in scattering matrix element computations
using the LSZ reduction formula. If the representation is regular for the Weyl elements 
we also have access to the fields themselves and not only their exponentials. Hence 
we specialise (\ref{3.4}) to
\be \label{3.5}
<\Omega_H,\; \rho(Q(f_N))(t_N)\;..\;\rho(Q(f_1))(t_1)\;\Omega_H>\;.
\ee   
This suggests to interpret the test functions $f_k(x)$ on $\sigma$ as instantaneous 
evaluations of a spacetime test function $F$ such that $F(t_l,x)=f_l(x)$ and to 
consider 
\be \label{3.6}
<\Omega_H,\; T_l(e^{i \rho(Q(F))})\;\Omega_H>,\;
\rho(Q(F))=\int\; d^{D+1}x\; \rho(Q_{ab}(x))(t) F^{ab}(t,x)\;,
\ee  
where the time ordering symbol $T_l$ instructs to order the latest time dependence to the 
outmost left. The functional derivatives of (\ref{3.6}) with respect to $F$ at $F=0$ yield the 
time ordered $N-$point functions. 

In trying to derive a practically useful expression for (\ref{3.6}) e.g. in form of a 
path integral, an inconvenient fact is that (\ref{3.6}) involves 
the vector $\Omega_H$ which is typically not explicitly known. We assume that the time evolution 
is {\it mixing} \cite{17} (a stronger condition than ergodicity), that is, for any 
$\psi_1,\psi_2\in {\cal H}$
\be \label{3.7}
\lim_{T\to \infty} <\psi,U(-T)\psi'>=<\psi_1,\Omega_H>\; <\Omega_H,\psi_2>
\ee
(the same relation then also holds for $T\to-\infty$). Suppose that $F$ has compact 
time support in $[-T,T]$. We define $t_l=l/N\; T,\;l=-N,..,N-1;\;\Delta=T/N$ and $f_l(x)=F(t_l,x)$.
Then (\ref{3.6}) can be written
\be \label{3.8}
\lim_{T\to \infty} \;\lim_{N\to \infty}\; <\Omega_H,U(T)\; e^{i\Delta \rho(Q(f_{N-1}))}
e^{i\Delta \rho(H)}..e^{i\Delta \rho(H)}\;e^{i\Delta \rho(Q(f_{-N})}\; U(T)\Omega_H>\;,
\ee
We extend (\ref{3.8}) by $<\Omega,\Omega_H>$ and use (\ref{3.7}) with $\psi_1=\Omega$ and $\psi_2$ the 
ket of (\ref{3.8}) to obtain for (\ref{3.8})
\be \label{3.9}
\lim_{T\to \infty} \;\lim_{N\to \infty}\; \frac{1}{<\Omega,\Omega_H>}\;
<\Omega,\;e^{i\Delta \rho(Q(f_{N-1}))}
e^{i\Delta \rho(H)}..e^{i\Delta \rho(H)}\;e^{i\Delta \rho(Q(f_{-N})}\; U(T)\Omega_H>\;.
\ee
We extend (\ref{3.9}) by $<\Omega_H,\Omega>$ and apply (\ref{3.7}) with 
$\psi_1=[e^{i\Delta \rho(Q(f_{N-1}))}
e^{i\Delta \rho(H)}..e^{i\Delta \rho(H)}\;e^{i\Delta \rho(Q(f_{-N})}\; U(T)]^\dagger\Omega$
and $\psi_2=\Omega$. Then (\ref{3.9}) becomes
\be \label{3.10}
\lim_{T\to \infty} \;\lim_{N\to \infty}\; \frac{1}{|<\Omega,\Omega_H>|^2}\;
<\Omega,\;e^{i\Delta \rho(Q(f_{N-1}))}
e^{i\Delta \rho(H)}..e^{i\Delta \rho(H)}\;e^{i\Delta \rho(Q(f_{-N})}\;\Omega>\;.
\ee
Finally, using again (\ref{3.7}), we can substitute the denominator so that (\ref{3.10}) becomes
\be \label{3.11}
\lim_{T\to \infty} \;\lim_{N\to \infty}\; \
\frac{
<\Omega,\;e^{i\Delta \rho(Q(f_{N-1}))}
e^{i\Delta \rho(H)}..e^{i\Delta \rho(H)}\;e^{i\Delta \rho(Q(f_{-N})}\;\Omega>
}
{
<\Omega,U(-2T)\Omega>
}\;.
\ee    
This formula no longer refers to $\Omega_H$. To finally obtain a path integral formulation
one inserts resolutions of unity and relies on Feynman -- Kac type of arguments. 
As it stands, this would lead to the Feynman path integral
for time ordered functions. Instead, we pass to the Euclidean formulation and analytically 
continue $T\to iS$. This yields with $s_l=l S/N, \Delta'=S/N, F'(s_l,x):=-F(i\;s_l,x)$
\be \label{3.12}
\chi[F']=\lim_{S\to \infty} \;\lim_{N\to \infty}\; 
\frac{
<\Omega,\;e^{\Delta' \rho(Q(F'(s_{N-1})))}
e^{-\Delta' \rho(H)}..e^{-\Delta' \rho(H)}\;e^{\Delta' \rho(Q(F'(s_{-N}))}\;\Omega>
}
{
<\Omega,e^{-2 S H}\Omega>
}\;.
\ee   
Note that (\ref{3.12}) still contains the time zero fields as operators. Therefore 
the Wick rotation performed just affects the Heisenberg evolution. The time zero 
fields remain untouched by the Wick rotation and therefore never develop a non-analytic time dependence.  
Also note that passing to the Euclidean formulation and obtaining Schwinger functions rather 
than Feynman functions as their analytic continuation is especially attractive if 
$H$ is bounded from below so that $e^{-\Delta' \rho(H)}$ is a bounded operator which 
also improves the convergence of the corresponding path integral. This 
is not the case for GR. Rather our motivation to pass to the Euclidean formulation is 
that it is this formulation that is favoured in the ASQG approach.   

The subsequent relations can be properly justified only by compactifying not only 
time (by $S$) but also $\sigma$
(IR cutoff) and by discretising the fields not only temporally (by $N$) but also 
spatially on a lattice (UV-cutoff) so that 
the number of degrees is finite. In this case the representation $(\rho,{\cal H})$ is 
necessarily unitarily equivalent to the Schr\"odinger representation when irreducible and 
regular and we can sandwich resolutions of the identity in between the various factors appearing in (\ref{3.12})
both in numerator and denominator. This step is standard and we just note the formal end result
after having removed all regulators (we drop the primes and relabel $s$ by $X^0$ 
and consider spatial spatime coordinates $X^a:=x^a$  so that $X=(s,x)$)
\be \label{3.13}
\chi[F]=\frac{Z[F]}{Z[0]},\;\;
Z[F]=\int\; d\mu_L[Q,P]\; e^{-\int_{\mathbb{R}\times\sigma}\; d^{D+1}X\;[H(X)+iP^{ab}(X)\dot{Q}_{ab}(X)]} 
\exp{<Q,F>}\;\Omega[Q(\infty)]^\ast\; \Omega[Q(-\infty)]
\ee
Here $<Q,F>=\int\; d^{D+1}X\; F^{ab}(X) Q_{ab}(X)$, $\Omega[Q]$ is the Schr\"odinger 
representation of the cyclic vector $\Omega$ as a functional of the time zero field $Q$
and $d\mu_L=\prod_{x^0\in \mathbb{R}}\; [dQ(x^0)]\;[dP(x^0)]$ with
$[dQ(x^0)]=\prod_{x\in \sigma}[d^k Q(x^0,x)], \; [dP(x^0)]=\prod_{x\in \sigma}[d^k P(x^0,x)/(2\pi)^k],\; k=D(D+1)/2$ 
is the formal Liouville measure (one for each spacetime point). It is important to understand that the time 
dependence of the integration variables simply comes from the insertions of unity 
$1_{{\cal H}} =\int\; [dQ] \delta_Q <\delta_Q,.>=\int\; [dP] e_P <e_P,.>$
at the 
various times and we label the integration variables of these resolutions of unity 
by the time slot at which we insert them. The factor $i$ in the exponent is no mistake, it comes 
from the position and momentum eigenfunctions $\delta_Q[Q']=\delta[Q,Q'],\; 
e_P[Q']=e^{-i\int_\sigma\;d^Dx \; P^{ab} Q'_{ab}}$. All of this is exactly the same as 
in the Feynman path integral, the only difference is that the Hamiltonian density 
$H(X)=H(Q(X),P(X))$ does not come with a pre-factor of $i$ in the exponent since we did not
sandwich between Weyl elements and  
unitary operators $e^{i\Delta t \;\rho(H)}$ but rather contraction operators $e^{-\Delta s \rho(H)}$.

The next step is to integrate over the momenta. This can be done pointwise 
in spacetime and we can immediately write the result using the preparations of the previous section
(we use $M=\mathbb{R}\times \sigma$) 
\ba \label{3.14}
\chi[F] &=& \frac{Z[F]}{Z[0]},\;\;
\nonumber\\
Z[F] &=&\int\; [dQ]\; e^{\kappa^{-1}\int_M\; d^{D+1}X\;\sqrt{\det(g(Q))}[R(g(Q))-2\Lambda]} \;
\exp{<Q,F>}\;\times 
\nonumber\\
&& \Omega[Q(\infty)]^\ast\; \Omega[Q(-\infty)]\; e^{-2[\dot{V}(q(Q(\infty)))-\dot{V}(q(Q(-\infty)))]}\;,
\ea
where $Q$ independent numerical factors have cancelled between numerator and denominator. Here $q$ is 
constructed from $Q$ via $Q=[\det(q)]^r\; q,\; r=\frac{D-4}{4D}$ and $g(Q)$ is the Euclidean signature 
spacetime metric constructed from $q(Q)$ with $X$ considered as synchronous coordinates, i.e.
we have the Euclidean line element $g_{\mu\nu}(X)\; dX^\mu\;dX^\nu=(dx^0)^2+[q(Q))]_{ab}(X)\;dx^a\; dx^b$.
The value $r$ was chosen specifically in order that the integral over momenta produces only 
a $Q$ independent Jacobian which cancels between numerator and denominator. Finally, $V(q(x^0)$ is the total
volume of $\sigma$ at $x^0$.\\
\\
In what follows the dependence on cyclic vector in the distant 
past and future will play no role while we keep the Gibbons-Hawking boundary 
term that forces the Euclidean action to be in the canonical form (\ref{2.29}). Note that the exponential of the Euclidean action 
$S=+\kappa^{-1}\int d^{D+1}X |\det(g)|^{1/2}[R(g)-2\Lambda]$ has the correct (positive) sign as the 
kinetic term enters with a minus sign into $R(g)$ for Euclidean signature. Still the integral
(\ref{3.14}) is not granted to converge even for positive $\Lambda$ since the spatial Ricci
scalar is indefinite and the kinetic term contains a negative ``conformal mode'' \cite{19} (the DeWitt metric 
has signature $(-1,+1,..,+1)$). For the same reason, it is unclear whether $Z(F)$ is time 
reflection positive \cite{8}, a minimal requirement in order to regain 
${\cal H},\Omega_H, H$ via Osterwalder-Schrader reconstruction as the latter necessarily 
produces a Hamiltonian operator bounded from below and our Hamiltonian does not obviously 
have this property. See \cite{21} for a discussion of reflection positivity in Euclidean 
quantum gravity. The absence of reflection positivity does not mean that there is 
no underlying Hilbert space structure and Hamiltonian, just that the Hamiltonian is not 
bounded from below and thus Osterwalder-Schrader reconstruction cannot be used.

\section{Laplacians, heat kernels, cutoffs and Wetterich equation}
\label{s4}

In the first subsection we summarise the properties of the generating functional 
of Schwinger functions obtained. In the second we analyse the symmetries of the Euclidean action 
which influences the choice of cutoff functions used to construct the 
effective average action and study the relation between $(D+1)$- and $D$-dimensional tensors
densities. In the third we use that relation to synthesise a natural projector 
that necessarily finds its way into suitable Laplacians on our theory space. In the fourth 
we define the Effective average action and quickly comment on how to generalise 
the usual heat kernel expansion to non-trivial density weight.  

\subsection{Starting point}
\label{s4.1}

Since at the end of the previous section we ended up with a path integral involving 
the exponential of the Einstein--Hilbert action for Euclidean signature, it 
is worthwhile to list what has been gained as compared to the usual approach to
ASQG:
\begin{itemize}
\item[1.] {\it Connection to Operator Formulation}\\
The path integral was not written down ``by analogy'' but was derived from the 
language of operators and Hilbert spaces of {\it Lorentzian signature} GR.
\item[2.] {\it Physical interpretation and observables}\\
The Gaussian dust matter selects a natural reference frame and therefore a natural 
notion of time with corresponding Hamiltonian $H$. The ``problem of time'' is naturally 
solved and the observables of the theory as measured in this reference frame are 
cleanly identified. The gauge is fixed prior to quantisation and  
constructions to deal with gauge redundancies (such as ghost integrals implementing 
Faddeev--Popov determinants) never enter the stage.\footnote{See \cite{22a} for very recent developments towards a  construction of a gauge invariant effective action.} The quantum field theory 
to be constructed is that for an ordinary, albeit highly non-linear, Hamiltonian 
system with conservative Hamiltonian that allows for an infinite number of conserved charges
and resembles a non-linear $\sigma$ model. 
\item[3.] {\it Euclidean Einstein--Hilbert action for physical Lorentzian quantum GR}\\
While the Hamiltonian operator is for physical, Lorentzian signature GR, its 
Heisenberg time evolution by unitary operators $e^{it H}$ that enters the time ordered 
N-point functions has a natural analytic continuation $t\to is$ to the self-adjoint 
operators $e^{-s H}$ (which would be contractions if the spectrum of $H$ is bounded from
below) because we do have a natural notion of time at our disposal. In this way 
we arrive at the Schwinger functions of the theory, i.e. its Euclidean formulation. 
\item[4.] {\it Measure Jacobian}\\
The path integral formulation of the generating functional of (connected) Schwinger
functions in a first step leads to a functional integral over phase space rather 
than configuration space. Integrating out the momenta is possible as they 
enter quadratically in $H$, but they produce a non-trivial Jacobian that involves 
the determinant of the DeWitt metric. An often applied method to bring that Jacobian
into the exponent is to use Berezin integrals involving ghost fields. In this paper 
we chose a more direct route and performed a canonical transformation on the phase 
space from usual ADM variables $(q,p)$ to density valued canonical coordinates
$(Q,P)$ prior to quantisation to the effect to render that Jacobian trivial. 
\item[5.] {\it Euclidean action}\\    
The final configuration space path integral formulation of the generating functional of 
(connected) Schwinger functions indeed involves the exponential of the Euclidean signature 
Einstein--Hilbert action (plus Gibbons Hawking boundary term, see \cite{GH} for a treatment in ASQG of this term),
however, with two restrictions: 1. We only integrate over Euclidean signature 
metrics $g=g(q,N,\vec{N})$  
with fixed unit lapse $N=1$ and zero shift $\vec{N}=0$ 
as a consequence of having solved all gauge redundancies 
prior to quantisation and 2. the natural integration variable is the density valued 
field $Q$. Therefore the Euclidean signature Einstein--Hilbert action $S[g(q,N,\vec{N})]$ needs 
to be written in terms of these data, i.e. $S[Q]=S[g(q(Q),1,0)]$. This is the way to read the 
end result (\ref{3.14}) of the previous section.
\item[6.] {\it Analytically extended metrics}\\
It is often criticised that time Wick rotation of the path integral in quantum gravity  
is meaningless using the following argument: In Minkowski space, Wick rotation 
just changes the signature of a metric that has constant (time independent) tensor 
components. However in quantum gravity the metric, understood as an integration 
variable in the path integral, is generically time dependent, hence naively 
analytic continuation makes it complex valued rather than a real metric with 
Euclidean signature which would seemingly results in a path integral over complexified 
gravity. Indeed this is what would happen if the metric field would be a generic,
time dependent background field. However, in quantum gravity where the metric is 
a dynamical field operator valued distribution, this is actually not what happens,
at least when the Hamiltonian is not explicitly time dependent: in the canonical (operator) approach,
the time dependence
of the Lorentzian signature quantum metric field comes entirely from the unitary Heisenberg evolution of the 
self-adjoint time zero metric fields which keeps the time evolved quantum field self-adjoint. 
What happens upon Wick rotation is that indeed this time evolved quantum field operator is no longer self-adjoint.
However, what enters the Schwinger N-point function is a product of self-adjoint operators, 
consisting of time zero quantum fields and exponentials of the quantum Hamiltonian times a real
number. When cast into the form of a functional integral, each factor of a time zero quantum metric 
field loses its status as an operator and rather becomes an independent, real valued integration variable.
These integration variables are labelled by the real Euclidean time parameter corresponding 
to the Euclidean point of time at which it occurs in the product. This is why only real valued 
metrics enter the Euclidean path integral. Formally, when comparing the generating functionals of 
time ordered and Schwinger functions this corresponds to a switch from real to imaginary 
lapse and thus real to imaginary extrinsic curvature without touching the time dependence 
of the integration variables. This is quite similar to what is considered in ASQG \cite{3, 3b} 
and causal dynamical triangulations \cite{22}. 
\end{itemize}

\subsection{Interplay between symmetries and cutoff functions}
\label{s4.2}

We note that the generating functional of Schwinger functions 
\be \label{4.1}
\chi[F]=\frac{Z[F]}{Z[0]}, Z[F]=\int\; [dQ]\; e^{S[Q]}\; e^{<F,Q>}
\ee
is not covariant under all spacetime diffeomorphisms $\phi\in $Diff$_{D+1}(\mathbb{R}\times \sigma)$
but only those that preserve the 
synchronous gauge
\be \label{4.2}
[\phi^\ast g]_{t\mu}(X)=
\phi^\nu_{,t}\phi^\rho_{,\mu} g_{\nu\rho}(\phi(X))
=\phi^t_{,t}\phi^t_{,\mu}
+\phi^b_{,t}\phi^c_{,\mu} g_{bc}(\phi(X))=\delta^t_\mu\;.
\ee
This implies with $X=(t,x)$ that
$\phi^t_{,a}=-\frac{g_{bc}(\phi(X))\phi^b_{,t}\phi^c_{,a}}{\phi^t_{,t}}$ and 
$\phi^t_{,t}=\pm\sqrt{1-g_{bc} \varphi^b_{,t} \varphi^c_{,t}}$ when 
$\varphi^t_{,t}\not=0$. This system of $D+1$ PDE's for $\phi^t_{,\mu}$ has integrability conditions which 
depend on $g$. If we want the allowed class of $\phi$ to be independent of 
$g$ then we must set $\phi^a_{,t}=0$. Then $\phi^t_{,\mu}=\pm \delta^t_\mu$. The 
plus sign corresponds to the subgroup containing the identity, the minus sign 
refers to the coset containing time reflections. 
\begin{Definition} \label{def4.1}
The subgroup Diff$_D(\mathbb{R}\times \sigma)$ of the spacetime 
diffeomorphism group Diff$_{D+1}(\mathbb{R}\times \sigma)$ is isomorphic to a two-fold cover of the spatial  
diffeomorphism group Diff$(\sigma)$. Its elements are labelled by $\epsilon\in \{\pm 1\}$ and 
$\varphi\in$ Diff$(\sigma)$ and are explicitly given by 
\be \label{4.3}
\phi_{\epsilon,\varphi}(t,x):=(\epsilon t, \varphi(x))\;.
\ee
\end{Definition}
Note that only the subgroup of Diff$_D(\mathbb{R}\times \sigma)$ containing the identity
i.e. the diffeomorphisms $\phi_{+,\varphi}$ preserve the Hamiltonian while all of them
preserve the action displayed in (\ref{4.1})  In order not to get confused in what follows 
we distinguish between the following spaces of tensor fields on $\mathbb{R}\times \sigma$.
\begin{Definition} \label{def4.2} ~\\
i. $T_{D+1}(A,B,w)$ is the usual space of all smooth spacetime tensor fields of rapid decrease
on $\mathbb{R}\times \sigma$. 
That is, the index structure of an element $T$ is given by $T^{\mu_1 .. \mu_A}_{\nu_1 .. \nu_B}$ 
and under $\phi\in $ Diff$_{D+1}(\mathbb{R}\times \sigma)$ it transforms as 
\be \label{4.4}
[\phi^\ast T]^{\mu_1 .. \mu_A}_{\nu_1 .. \nu_B}(X)=|\det(J(X))|^w\;
\prod_{i=1}^A [J^{-1}(X)]^{\mu_i}_{\mu'_i}\;
\prod_{j=1}^B [J(X)]^{\nu'_j}_{\nu_j}\;T^{\mu'_1 .. \mu'_A}_{\nu'_1 .. \nu'_B}(\phi(X))\;,
\ee
where $J^\mu_\nu(X)=\frac{\partial\phi^\mu(X)}{\partial X^\nu}$ is the spacetime Jacobian 
of $\phi$ and 
$X=(t,x)$.\\
ii. $S_D(A,B,w)$ is the usual space of all tensor fields on $\sigma$ with an additional 
dependence on the time parameter $t$. Its elements $H$ are smooth and of rapid decrease 
with respect to both $t,x$ and carry the index structure $H^{a_1 .. a_A}_{b_1 .. b_B}$. 
Under $\phi_{\epsilon,\varphi}\in $ Diff$_D(\mathbb{R}\times \sigma)$ it transforms as 
\be \label{4.5}
[\varphi_{\epsilon,\varphi}^\ast H]^{a_1 .. a_A}_{b_1 .. b_B}(t,x)=|\det(J(x))|^w\;
\prod_{i=1}^A [J^{-1}(x)]^{a_i}_{a'_i}\;
\prod_{j=1}^B [J(x)]^{b'_j}_{b_j}\;H^{a'_1 .. a'_A}_{b'_1 .. b'_B}(\epsilon t, \varphi(x))\;,
\ee
where $J^a_b(x)=\frac{\partial\varphi^a(x)}{\partial x^b}$ is the spatial Jacobian 
of $\varphi$.
\end{Definition}
The relation between these spaces is as follows:
The space $T_{D+1}(A,B,w)$ is an irreducible representation of Diff$_{D+1}(\mathbb{R}\times \sigma)$.
Upon restriction to Diff$_D(\mathbb{R}\times \sigma)$ it decomposes into irreducible subspaces.
\begin{Lemma} \label{la4.1} ~\\
Upon restriction we have $T_{D+1}(A,B,w)=\oplus_{A'=0}^A\; \oplus_{B'=0}^B\; 
S_D(A',B',w)$. 
\end{Lemma}
That is, each subspace of $T_{D+1}(A,B,w)$ selected by fixing $k,l$ of its indices 
$\mu_i, \nu_j$ respectively to take the value $t$ is an invariant subspace and transforms as 
an element of $S_D(A-k,B-l,w)$ under the restricted diffeomorphism group. As a typical example
consider $T^\mu_\nu\in T_{D+1}(1,1,w)$. Then 
$T^t_t\in S_D(0,0,w),\;T^t_b\in S_D(0,1,w),\; T^a_t \in S_D(1,0,w),\;
T^a_b\in S_D(1,1,w)$ transform as spatial scalar, co-vector, vector and 2- tensor with 
weight $w$ respectively. This relies on the identity 
$|\det(\partial \phi_{\epsilon,\varphi}/\partial X)|=|\det(\partial \varphi/\partial x)|$.
The simple proof is left to the reader. 

We notice that the generating functional (\ref{4.1}) transforms under Diff$_D(\mathbb{R}\times \sigma)$
as $\chi[F]\mapsto \chi[(\phi_{\epsilon,\varphi}^{-1})^\ast F]$ where $F$ is considered 
an element of $S_D(2,0,1-2r)$ (since $Q$ has density weight $2r$). 
To see this note that the action is Diff$_D(\sigma)$ invariant and that the measure $[dQ]$ has 
a $Q$ independent Jacobian which drops from the quotient $Z[F]/Z[0]$. 

The formulation of an appropriate ASQG framework has to be adapted accordingly. Recall
that usually one employs the background field method to the Euclidean path integral 
and introduces an average kernel or cutoff that depends on a scale parameter $k$ and 
that background spacetime metric through its {\it spacetime Laplacian}. This is motivated by 
the fact that in the usual approach one assumes (rather than derives) that the path integral 
depends on the exponential of the Einstein--Hilbert action (and higher derivative Diff($M$)
invariant terms built from the spacetime metric) plus a gauge fixing term plus a ghost term (which 
brings the Faddeev--Popov determinant between constraints and gauge fixing condition into
the exponent) and one integrates over Euclidean signature spacetime metrics. Thus in order 
that the flow only generates terms compatible with the symmetries of the Euclidean action 
one builds the cutoff in a Diff($M$) invariant fashion (here $M=\mathbb{R}\times \sigma$). 
In the present case, we have only 
the fields $Q_{ab}, \bar{Q}_{ab}$ at our disposal. We could 
in fact complete $Q_{ab}$ to a spacetime field and then apply those types of usual 
cutoffs. However, this has several caveats. To see this, we define 
$g_{0\mu}(Q):=\delta^t_\mu$ and $g_{ab}(Q):=[\det(Q)]^{-\frac{r}{1+rD}}\; Q_{ab}$
and similar for $\bar{Q}$. Furthermore we can construct $G_{\mu\nu}(Q):=[\det(g(Q))]^r g_{\mu\nu}(Q)$. 
Let $\overline{\nabla}$ be the covariant differential compatible with $g(\bar{Q})$. Then 
it is easy to see that $(\overline{\nabla} G)_{t \mu}\not\propto \delta_\mu^t$, i.e. the covariant 
differential maps out of the space of allowed spacetime tensors. Furthermore, for 
$r\not=0$ the fluctuation $G_{t\mu}-\bar{G}_{t\mu}=
([\det(Q)]^{-\frac{r}{1+rD}}-[\det(\bar{Q})]^{\frac{r}{1+rD}})\;\delta^t_\mu$ is no longer 
a linear function of $H_{ab}=Q_{ab}-\bar{Q}_{ab}$ which however is an essential requirement 
that the cutoff function needs to have for the Wetterich equation to be valid.

It transpires that we need a different type of cutoff function that is adapted to 
the symmetries of the given Euclidean action which is just Diff$_D(\mathbb{R}\times \sigma)$. There are 
at least two natural options. The first option that was followed in \cite{3} is 
to construct a function just from the spatial Laplacian 
$\bar{q}^{ab}\; \bar{D}_a\; \bar{D}_b,\; \bar{D}_a \bar{q}_{bc}=0,\;
\bar{q}=[\det(\bar{Q})]^{-\frac{r}{1+rD}}\;\bar{Q}$. This has the disadvantage that 
the cutoff just controls the spatial fluctuations of the field. The second option 
that we will follow below is based on the following simple observation.
\begin{Lemma} \label{la4.2} ~\\
The operator $\bar{D}_t$ defined by 
$(\bar{D}_t T)^{a_1..a_A}_{b_1..b_B}:=\frac{\partial T^{a_1..a_A}_{b_1..b_B}(t,x)}{\partial t}$
preserves $S_D(A,B,w)$. 
\end{Lemma}
The proof is trivial as $J(x)$ is independent of $t$. It follows that $\bar{D}_t$ 
is to be considered a scalar operator on $S_D(A,B,w)$. We may therefore use also $D_t$
in oder to construct a Diff$(\mathbb{R}\times \sigma)$ invariant cutoff function. For instance we may 
use the spacetime Laplacian 
\be \label{4.3a}
\overline{\Delta}:=\bar{D}_t^2+\bar{q}^{ab} \bar{D}_a \bar{D}_b =
\bar{g}^{\mu\nu}\; \bar{D}_\mu\;\bar{D}_\nu 
\ee
with $\bar{g}_{t\mu}=\delta^t_\mu\;, \bar{g}_{ab}=\bar{q}_{ab}$. 
Note that $\overline{\Delta}$ is {\it not} $\overline{\Delta}_{D+1}:=\bar{g}^{\mu\nu} \bar{\nabla}_\mu \bar{\nabla}_\nu$
simply because $\bar{\nabla}_\mu\not=\bar{D}_\mu$, in particular $\bar{D}_\mu \bar{g}_{\nu\rho}\not=0$. 
We may call $\bar{D}_\mu$ the {\it hybrid covariant differential}. We are interested in the
{\it hybrid curvature tensor}.
\begin{Lemma} \label{la4.3} ~\\
Let $\bar{k}_{ab}:=\frac{1}{2}\bar{D}_t \bar{q}_{ab}$ be the the extrinsic curvature of $\bar{q}$ and 
$\bar{\Gamma}^c_{ab}$ the Christoffel symbol of $\bar{q}$. Then 
\be \label{4.6}
C^c_{ab}:=\partial_t \bar{\Gamma}^c_{ab}=2\bar{D}_{(a} \bar{k}_{b)}\;^c-\bar{D}^c\;\bar{k}_{ab}\;,
\ee
where index transport is with respect to $\bar{q}_{ab}$.
\end{Lemma}
The computation is standard. While $\bar{\Gamma}$ is not a tensor field, its variation is a tensor 
field as explicitly displayed by the r.h.s. of (\ref{4.6}). Let now $H\in S_D(A,B,w)$. 
Then $[\bar{D}_a,\bar{D}_b] H$ can be expressed in the standard way in terms of the Riemann tensor
$\bar{R}_{abcd}$ of $\bar{q}$. On the other hand:
\begin{Lemma} \label{la4.4} ~\\
We have
\be \label{4.7}
([\bar{D}_t, \bar{D}_c] H)^{a_1..a_A}_{b_1..b_B}
=\sum_{i=1}^A\; C^{a_i}_{cd}\; H^{a_1..d..a_A}_{b_1..b_B}
-\sum_{j=1}^B\; C^{d}_{c b_j}\; H^{a_1..a_A}_{b_1..d..b_B}
-w\;C^d_{cd}\; H^{a_1..a_A}_{b_1..b_B}\;.
\ee
\end{Lemma}
The proof consists of applying the standard formula for $\bar{D}_a T$ and applying 
(\ref{4.6}). It follows that the spatial-spatial curvature of the hybrid differential 
is determined by the 
{\it Riemannian curvature} of $\bar{q}$ while the temporal-spatial curvature is determined by 
the {\it extrinsic curvature} of $\bar{q}$. This is quite appealing for it means that the 
renormalisation flow precisely generates those terms which are already part of the 
Euclidean action. We anticipate that the exact flow, as defined below, generates an effective 
average action $\bar{\Gamma}_k[\hat{Q},\bar{Q}]$ such that the actual effective action 
$\Gamma[\hat{Q}]$ of our quantum field theory defined as the Legendre transform of $\ln(Z[F])$     
and given by $(\bar{\Gamma}_k[\hat{Q}',\bar{Q}])_{k=0,\hat{Q}'=0,\bar{Q}=\hat{Q}}$ is the most 
general Diff($\sigma$) invariant functional that one can build from $\hat{q}_{ab}, R_{abcd}(\hat{q})$ and
the $\hat{D}_t, \hat{D}_a$ derivatives thereof, where $\hat{q}=[\det(\hat{Q})]^{-\frac{r}{1+rD}} \hat{Q}$.
The simplest such terms not containing higher time derivatives (which would generate an 
Ostrogradsky instability \cite{24}) 
are (dropping the hat, using 
$(R_4)_{abcd}=R_{abcd},\; (R_2)_{ab}=q^{cd} R_{acbd},\; R_0= q^{ab} R_{ab}$ and traces are to 
be formed using $q$)
\be \label{4.8}
\sqrt{\det(q)}\{[R_0]^n,\; [{\sf Tr}([R_2]^m)]^n,\; [{\sf Tr}([R_4]^m)]^n,\; 
[{\sf Tr}(k^m)]^n,\; [{\sf Tr}([R_2\cdot k]^m)]^n,\; [{\sf Tr}([D R_2]^m)]^n
[{\sf Tr}([D k]^m)]^n,..\}\;,
\ee
where $D$ acts only into the spatial direction. 

It is conceivable that a minimal list of such 
terms to close the flow is downsized by the requirement such terms are to arise as the 
specialisation of a Diff$_{D+1}(M)$) invariant term to synchronous coordinates. Indeed 
while (\ref{4.3a}) is certainly a possible choice as far as the Diff$_D(M)$ covariance 
is concerned, the following list of requirements has to be met by an admissible $\overline{\Delta}$:\\
1.\\ 
$\overline{\Delta}$ preserves the real vector space $S_D(A,B,w)$. 
This makes sure that $\Delta$ does not map out of the given theory space.\\
2. \\
$\overline{\Delta}$ is a negative semi-definite (and therefore symmetric) operator with respect to the 
inner product on $S_D(A,B,w)$ defined by
\be \label{4.9}
<H,\hat{H}>_D:=\int_{\mathbb{R}}\; dt\; \int_\sigma\; d^Dx\; 
\sqrt{\det(\bar{q})}^{1-2w}\;\prod_{i=1}^A\; \bar{q}_{a_i a'_i}\;
\prod_{j=1}^B\; \bar{q}^{b_j b'_j}\; H^{a_1..a_A}_{b_1 .. b_B}\;
\hat{H}^{a'_1..a'_A}_{b'_1 .. b'_B}\;.
\ee
This makes sure that we can perform useful functional analysis with $\bar{\Delta}$.
Note that the inner product is positive definite and Diff$(\sigma)$ invariant if both
$\bar{q}$ and $H,\hat{H}$ transform according to their indicated tensor density type
with respect to Diff$_D(M)$.\\ 
3.\\ 
$\overline{\Delta}$ reduces to the flat space Laplacian when $\bar{q}_{ab}=\delta_{ab}$.
This ensures that cutoff functions constructed from the flat space $\overline{\Delta}$ have the same 
analytical properties when $\overline{\Delta}$ is generalised to curved space.\\  
\\
It is clear that (\ref{4.3a}) obeys requirements 1. and 3. but violates 2. 
However, it is easy to construct infinitely many $\overline{\Delta}$ that obey all three requirements
by backwards engineering: Consider the manifestly negative semidefinite 
object (we subsume $a:=(a_1..a_A), \; b=(b_1..b_B)$ 
into compound spatial indices and write likewise $\bar{q}_{aa'}=\prod_i \bar{q}_{a_i a'_i}$ etc.)
\ba \label{4.10}
<H,\; \overline{\Delta} \hat{H}>&:=&
-\int\; d^{D+1}X\; \sqrt{\det(\bar{q})}\; \{(D_t H)^a_b\; q_{aa'}\; q^{bb'}\;(D_t \hat{H})^{a'}_{b'}
+(D_c H)^a_b\; q_{aa'}\; q^{bb'}\;q^{cc'}\;(D_{c'} \hat{H})^{a'}_{b'}
\nonumber\\
&& +\sum_{n=1}^\infty\; \kappa_n\; ([\bar{k}^n]_a^b\; H^a_b)\;([\bar{k}^n]_{a'}^{b'}\; H^{a'}_{b'})\}\;,
\ea
where $\bar{k}^n$ is the $n-$th power of $\bar{k}_a^b=\bar{k}_{ac} \bar{q}^{cb}$ and $\kappa_n\ge 0$ and non vanishing 
for finitely many $n$ only. We read off $\bar{\Delta}$ using simple 
integration by parts exploiting the rapid decrease assumption. E.g. for $\kappa_n=0\; \forall
\; n$
\be \label{4.11}  
[\overline{\Delta} H]^a_b =\sqrt{\det(\bar{q})}^{-[1-2w]}\; \bar{q}^{aa'}\; \bar{q}_{bb'}\; (\bar{D}_t(\sqrt{\det(\bar{q})}^{1-2w}\; 
\bar{q}_{\cdot c'} \bar{q}^{\cdot d'} (\bar{D}_t H)^{c'}_{d'}))_{a'}^{b'}
+(\bar{q}^{cd} \bar{D}_c \bar{D}_d H)^a_b\;.
\ee
Since in the flat space limit we have $\bar{k}=0$ it follows that all of (\ref{4.8}) obey 1.-3.

\subsection{Projection structure and associated Laplacians}
\label{s4.3}

To downsize the number of possibilities and to tie the flow generated by $\overline{\Delta}$ to 
its spacetime origin we note the following:\\ 
Equip $T_{D+1}(A,B,w)$  with the inner product
\be \label{4.12}
<T,\hat{T}>_{D+1}
:=\int_{\mathbb{R}}\; dt\; \int_\sigma\; d^Dx\; 
\sqrt{\det(\bar{g})}^{1-2w}\;\bar{g}_{\mu\mu'}\;\bar{g}^{\nu\nu'}
T^\mu_\nu\; \hat{T}^{\mu'}_{\nu'}
\ee
This inner product is positive definite and Diff$_{D+1}(M)$ invariant if 
we let Diff$_{D+1}(M)$ act on all fields $\bar{g}, T, \hat{T}$. It is therefore in 
particular invariant under the subgroup Diff$_D(M)$. Consider now the following 
objects
($\mu=(\mu_1,..,\mu_A),\; \nu=(\nu_1, .. \nu_B)$ are compound spacetime indices
and $\delta^a_\mu=\prod_{i=1}^a \delta^{a_i}_{\mu_i}$ etc.) 
\be \label{4.13}
E^{\mu b}_{\nu a}=\delta^\mu_a\;\delta^b_\nu,\; 
R^{a\nu}_{b \mu}=\delta_\mu^a\;\delta_b^\nu
\ee
\begin{Lemma} \label{la4.5} ~\\
i.\\
The map 
\be \label{4.14}
E:\; S_D(A,B,w)\to T_{D+1}(A,B,w);\;\; [E\cdot H]^\mu_\nu:= E^{\mu b}_{a \nu} \; H^a_b
\ee
is an isometric embedding with respect to the Hilbert structures $<.,.>_D$ and 
$<.,.>_{D+1}$ respectively when $\bar{g}_{t\mu}=\delta^t_\mu$. \\
ii.\\
The adjoint of $E$ is the map 
\be \label{4.15}
R=E^\ast:\; T_{D+1}(A,B,w)\to S_D(A,B,w);\;\; [R\cdot T]^a_b:= R^{a \nu}_{b \mu} \; T^\mu_\nu
\ee
iii.\\
The image $T_D(A,B,w):=E\cdot S_D(A,B,w)$ is a Diff$_D(M)$ invariant subspace of 
$T_{D+1}(A,B,w)$ and we have 
\be \label{4.16}
R\cdot E={\sf id}_{S_D(A,B,w)},\; P:=E\cdot R: T_{D+1}(A,B,w)\to T_D(A,B,w)
\ee
is an orthogonal projection.
\end{Lemma}
\begin{proof}:
i.\\
The claim is that 
\be \label{4.17}
<E\cdot H, E\cdot \hat{H}>_{D+1}=<H, \hat{H}>_D\;,
\ee
which is easily verified using 
$\det(\bar{g})=\det(\bar{q})$ and block diagonality of $\bar{g}$, i.e. 
\be \label{4.18} 
\bar{g}_{\mu\nu} \delta^\nu_b=\delta^a_\mu\; \bar{q}_{ab},\;
\bar{g}^{\mu\nu} \delta_\nu^b=\delta_a^\mu\; \bar{q}^{ab}\;. 
\ee
ii.\\
The claim is that 
\be \label{4.19}
<T, E\cdot H>_{D+1}=<R\cdot T, H>_D
\ee
which again follows from block diagonality.\\
iii.\\
Invariance is by construction and 
a simple calculation based on $\delta^a_\mu \delta^\mu_b=\delta^a_b$ shows that $R\cdot E \cdot H=H$.
This implies $P^2=(E\cdot R)^2=E\cdot (R\cdot E)\cdot R=P$ and $P^\ast=R^\ast \cdot E^\ast
=E\cdot R=P$. 
\end{proof}
Let now 
\be \label{4.20}
\overline{\Delta}_{D+1}:=\bar{g}^{\mu\nu}\; \bar{\nabla}_\mu \bar{\nabla}_\nu
\ee
be the standard spacetime Laplacian with respect to $\bar{g},\; \bar{\nabla} \bar{g}=0$ and 
$\bar{g}_{t\mu}=\delta_\mu^t,\; \bar{g}_{ab}=\bar{q}_{ab}$.  Then 
$\overline{\Delta}_{D+1}$ preserves $T_{D+1}(A,B,w)$ and is negative definite with respect 
to (\ref{4.12}). To make $\overline{\Delta}_{D+1}$ act on our space of fields $S_D(A,B,w)$ we must 
first embed it via $E$ into $T_{D+1}(A,B,w)$.
However, $\overline{\Delta}_{D+1}$ does not preserve the subspace $T_D(A,B,w)$: It is 
easy to see that $[\overline{\Delta}_{D+1}\; E\cdot H]^\mu_\nu$ is generically not vanishing when some 
of the $\mu_i,\nu_j$ take the index value $t$. The above 
developments however suggest to define
\be \label{4.21}
\overline{\Delta}_D:=R\cdot \overline{\Delta}_{D+1} \cdot E\;,
\ee
which does preserve $S_D(A,B,w)$. Using $E=P\cdot E,\; R=R\cdot P$ an equivalent definition 
is 
\be \label{4.22}
\overline{\Delta}_D:=R\cdot \overline{\Delta}^P_{D+1} \cdot E,\;
\overline{\Delta}^P_{D+1}:=P\cdot \overline{\Delta}_{D+1} \cdot P\;,
\ee
where $\bar{\Delta}^P_{D+1}$ is the projected Laplacian that preserves $T_D(A,B,w)$.
The advantage of $P$ over $E$ is that both $\overline{\Delta}_{D+1}, P$ act on 
$T_{D+1}(A,B,w)$. Due to $R^\ast=E$ and $P^\ast=P$ the operator $\overline{\Delta}_D$ on $S_D(A,B,w)$ is 
manifestly negative semi definite and thus 
symmetric on this respective domain. For instance
\ba \label{4.23}
&& <H,\overline{\Delta}_D\;\hat{H}>_D=<H,E^\ast\cdot  \overline{\Delta}_{D+1}\cdot E\cdot \hat{H}>_{D+1}
=<E\cdot H, \overline{\Delta}_{D+1}\cdot E\cdot \hat{H}>_{D+1}
\nonumber\\
&=&<\overline{\Delta}_{D+1}\cdot E\cdot H, E\cdot \hat{H}>_{D+1}
=<E^\ast \cdot \overline{\Delta}_{D+1}\cdot E\cdot H,\hat{H}>_D
\ea
and negative semi-definiteness of $\overline{\Delta}_D$ is inherited from $\overline{\Delta}_{D+1}$.

Note that despite the notation, $\overline{\Delta}_D$ is second order with respect to both 
$D_t, D_a$. Applied to our concrete theory, it would now be natural to construct cutoff functions of the form
($H_{ab}=Q_{ab}-\bar{Q}_{ab}$ and $\bar{q}=[\det(\bar{Q})]^{-\frac{r}{1+rD}} \bar{Q}$)
\be \label{4.24}
R_k(H,\bar{Q}):=<H, R_k(-\overline{\Delta}_D) H>_D\;,
\ee
where 
\be \label{4.25}
R_k(z)=k^2\; r(z/k^2),\; r(y)=\int_0^\infty\; ds\; \hat{r}(s)\; e^{-s\;y}
\ee
is the Laplace transform of the function $\hat{r}(s)$ which is to be smooth and of rapid decrease 
in heat kernel time $s$ both as $s\to 0 $ and $s\to \infty$ in order to be useful for ASQG as argued in 
\cite{4,4a}. A typical example is $r(s)=e^{-s^2-s^{-2}}$. Accordingly 
\be \label{4.26}
R_k(H,\bar{Q}):=k^2\; \int_0^\infty\; ds \; \hat{r}(s) <H, e^{s\; \overline{\Delta}_D/k^2}\; H>_D\;,
\ee
which involves the heat kernel of $\overline{\Delta}_D$. 
This looks now almost as in the standard case, except that 
$\overline{\Delta}_{D+1}$ is replaced by $\overline{\Delta}_D$ and the backgrounds on which 
$\overline{\Delta}_{D+1}$ is based is restricted to be in synchronous gauge. This innocent 
looking modification bears however the following technical nuisance: The standard 
heat kernel expansion techniques do not immediately apply. To se this we note that 
\be \label{4.27}
e^{s\overline{\Delta}_D/k^2}=E^\ast \cdot e^{s\overline{\Delta}^P_{D+1}} \cdot E\;,
\ee
where we used isometry $E^\ast E=$ id in the formal Taylor expansion of the exponential 
function. 
One would now like to develop either expansion techniques directly for 
$\overline{\Delta}^P_{D+1}$
or try to relate them to those for $\overline{\Delta}_{D+1}$. The problem with the first 
route is that $[\overline{\Delta}_{D+1},P]\not=0$ as already indicated above. See appendix
\ref{sa} for details. This fact implies 
that all formulae using the Synge world function heavily used in heat kernel expansions 
have to be rederived, perhaps using the technology developed for Horava--Lifshitz 
gravity \cite{25a}
in \cite{25}. In the appendix we sketch a method that uses S-matrix perturbation theory 
where $s$ is the perturbation parameter, and the theory of non-minimal operators developed in \cite{26}.
At any finite perturbative $s$ order the non-minimal operators involved are polynomials in the iterated 
commutators $[\overline{\Delta},P]_n$, see appendix \ref{sa} for details. It is important 
to have formulated these corrections in terms of $P$ rather than 
$E$ as $E, \overline{\Delta}_{D+1}$ act on different spaces.

With respect to that perturbative scheme the zeroth order is given by the following simpler 
version of 
(\ref{4.24})
\be \label{4.28}
R_k(H,\bar{Q}):=<E\cdot H, R_k(-\overline{\Delta}_{D+1})\;E\cdot  H>_{D+1}\;,
\ee
which involves only the ``projection'' of the standard heat kernel $e^{s\overline{\Delta}_{D+1}}$
to the space $S_D(A=0,B=2,w=2r)$. 
The advantage is that we can now 
copy all the heat kernel machinery from the standard case without using complicated 
non-minimal operators. One may even argue that 
the heat kernel and not the Laplacian is the fundamental object and in that sense 
(\ref{4.28}) could be argued to be ``more natural'' than (\ref{4.24}) thus not 
taking the corrections into account. In this paper we will start with (\ref{4.28}) 
as a first step, keeping the non-minimal operator corrections for future treatment. 

\subsection{Effective average action and heat kernel expansion}
\label{s4.4}

The remaining steps are now standard. First we employ the 
background field method and replace $Z[F]$ by 
\be \label{4.29}
\bar{Z}[F,\bar{Q}]:=\int\; [dH] \; e^{S[\bar{Q}+H]}\; e^{<F,H>}
\ee
and then we introduce the cutoff kernel
\be \label{4.30}
\bar{Z}_k[F,\bar{Q}]:=\int\; [dH] \; e^{S[\bar{Q}+H]}\; e^{-\frac{1}{2}\;R_k(H;\bar{Q})}\; e^{<F,H>}\;.
\ee
Then 
\be \label{4.31}
\bar{C}_k(F,\bar{Q}):=\ln(\bar{Z}_k(F;\bar{Q}),\;\;
\bar{\Gamma}_k(\hat{Q},\bar{Q}):={\sf extr}_F[<F,\hat{Q}>-\bar{C}_k(F;\bar{Q})]-\frac{1}{2}R_k(\hat{Q},\bar{Q})\;.
\ee
By construction, (\ref{4.13}) obeys the {\it Wetterich identity}
\be \label{4.32}
\frac{\partial}{\partial k} \bar{\Gamma}_k(\hat{Q},\bar{Q})=
\frac{1}{2}{\sf Tr}([R_k+\bar{\Gamma}^{(2)}_k(\hat{Q},\bar{Q})]^{-1}\; [\partial_k R_k(.,\bar{Q})])\;,
\ee
where $\bar{\Gamma}_k^{(2)}(\hat{Q},\bar{Q})$ is the second functional derivative of 
$\bar{\Gamma}_k(\hat{Q},\bar{Q})$ with respect to $\hat{Q}$ understood as a bi-distribution and 
also $\partial_k R_k$ is understood as a symmetric bi-distribution. To evaluate this identity 
one Taylor expands both sides in $\hat{H}=\hat{Q}-\bar{Q}$ to the desired precision. In what 
follows we will be content with the zeroth order. Then one makes an Ansatz for 
$\bar{\Gamma}_k(\hat{Q},\bar{Q})$ involving finitely many terms of the type (\ref{4.8}) with 
$k$-dependent couplings on the l.h.s. and retains on the r.h.s. only those terms of the 
same type (truncation of theory space). In what follows we will be content with the 
Einstein--Hilbert truncation.

As a final remark, note that the heat kernel expansion of $e^{s\overline{\Delta}_{D+1}}$ acting 
on $T_{D+1}(A,B,w)$ is usually only considered for the case $w=0$. We may compute it for 
$w\not=0$ as follows: By definition, the heat kernel $K(s)$ is a bi-tensor  
of type $(A,B,w)$ at $X$ and of type $(B,A,-w)$ at $Y$ such that (we relabel $s/k^2$ by $s$)
\be \label{4.33}
[e^{s\overline{\Delta}_{D+1}}\cdot T]^\mu_\nu(X)=\int_M\; d^{D+1}Y\; \sqrt{\det(\bar{g}(Y))}\; 
K^{\mu\beta}_{\nu\alpha}(s;X,Y)\; T^\alpha_\beta(Y)\;,
\ee
where Greek letters from the beginning/middle of the alphabet refer to the tensor structure 
at $Y,X$ respectively. Hence
\be \label{4.34}
\frac{d}{ds} K^{\mu\beta}_{\nu\rho}(s;X,Y)=
[\overline{\Delta}_{D+1} K^{\cdot\beta}_{\cdot\rho}(s;\cdot,Y)]^\mu_\nu(X),\;
K^{\mu\beta}_{\nu\alpha}(0;X,Y)=\delta^\mu_\nu \delta^\beta_\alpha\; \delta_{\bar{g}}(X,Y)\;,
\ee
where $\delta_{\bar{g}}(X,Y)$ is the bi-scalar valued $\delta$ distribution related to 
the coordinate $\delta$ distribution by $\delta(X,Y)=\sqrt{\det(\bar{g}(Y))}\;\delta_{\bar{g}}(X,Y)$.
This is consistent with the kernel to be of density weight zero at coincident points. 
Denoting by $\sigma(X,Y)$ the Synge world function (half of the square of the geodesic 
distance between $X,Y$ with respect to $\bar{g}$) the Ansatz reads
\be \label{4.35}
K^{\mu\beta}_{\nu\alpha}(s;X,Y)
=\frac{1}{\sqrt{4\pi s}^{D+1}}\;
\frac{|\det(\bar{\nabla}^X\bar{\nabla}^Y \sigma)|}{\sqrt{\det(\bar{g}(X))\det(\bar{g}(y))}}\; e^{-\frac{\sigma(X,Y)}{2s}}\; 
[\frac{\det(\bar{g}(X))}{\det(\bar{g}(Y))}]^{w/2}\; \Omega^{\mu\beta}_{\nu\alpha}(s;X,Y)  \;,
\ee
where the first three factors form a bi-scalar (the second is the van Vleck -- Morette determinant) chosen such that in the limit 
$s=0$ these converge to $\delta_{\bar{g}}$ and at $s\not=0$ but flat $\bar{g}$ coincides with the 
scalar heat kernel on $\mathbb{R}^{D+1}$. The fourth factor is new and takes care of the 
density weight. The actual heat kernel expansion concerns the bi-tensor $\Omega(s)$ of 
type $(A,B,0)$ at $X$ and $(B,A,0)$ at $Y$. It is normalised such that 
$\Omega^{\mu\beta}_{\nu\alpha}(0,X,X)=\delta^\mu_\alpha\delta^\beta_\nu$. 
One now inserts (\ref{4.31}) into (\ref{4.30}) and obtains an exact PDE system for 
$\Omega(s)$. To solve it one expands $\Omega(s)=\sum_{n=0}\; s^n\; \Omega(n)$ where 
$\Omega(n)$ is still a bi-tensor of the same type as $\Omega(s)$ and obtains an iterative 
PDE scheme that expresses $\Omega_{n+1}$ in terms of $\Omega_m,\; m\le n$. To solve 
it and reduce the PDE system to an algebraic system we use the parallel propagators 
$h^\mu_\alpha(X,Y)$, i.e. the holonomies of the Christoffel connection along the geodesic 
from $X$ to $Y$ which are bi-tensors of type $(1,0,0)$ at $X$ and $(0,1,0)$ at $Y$.
We now write 
\be \label{4.36}
\Omega^{\mu\rho}_{\nu\sigma}(n;X,Y):=
\Omega^{\mu\beta}_{\nu\alpha}(n,X,Y)\; h^\rho_\beta(X,Y) h^\alpha_\sigma(Y,X)\;,
\ee
with inversion 
\be \label{4.37}
\Omega^{\mu\beta}_{\nu\alpha}(n;X,Y):=
\Omega^{\mu\rho}_{\nu\sigma}(n,X,Y)\; h^\sigma_\alpha(X,Y) h^\beta_\rho(Y,X)\;,
\ee
and use that (\ref{4.36}) is a bi-tensor of type $(A+B,A+B,0)$ at $X$ and type 
$(0,0,0)$ at $Y$. Being a scalar at $Y$, we may therefore Taylor expand with respect to 
the $Y$ dependence
\be \label{4.38}
\Omega^{\mu\rho}_{\nu\sigma}(n; X,Y)=
\sum_{m=0}^\infty \Omega^{\mu\rho\lambda_1..\lambda_m}_{\nu\sigma}(n,m;X)\;
[\bar{\nabla}_{\lambda_1}\sigma].. 
[\bar{\nabla}_{\lambda_m}\sigma]\;
\ee
where $\Omega^{\mu\rho\lambda_1..\lambda_m}_{\nu\sigma}(n,m;X)$ is an ordinary mono-tensor 
of type $(A+B+m,A+B,0)$ at $X$. Using the master equations that hold for 
$\sigma(X,Y),\; h^\mu_\alpha(X,Y)$ \cite{27}, these tensors can be 
iteratively determined by the coincidence limits $Y\to X$ of repeated covariant derivatives
$\bar{\nabla}_\mu$ at $X$ of $\Omega^{\mu\rho}_{\nu\sigma}(n;X,Y)$.

\section{Analysis of the Wetterich equation} \label{s5}

In this section we will run the RG methods developed in ASQG on our model. First of all, we will specify the  truncation we will carry our 
analysis on, together with the choice of the cutoff. Subsequently, we will write down the beta functions of Newton's coupling and the 
cosmological constant. We determine  the fixed points and the critical exponents. Finally, we integrate down to $k \to 0$ the flow equations, 
obtaining the $k = 0$ limit of the dimensionful couplings. We emphasise that compared to standard ASQG treatments, 
our methods allow for an integration down towards the deep infrared ($k = 0$), obtaining the 
physical effective action.

\subsection{Einstein--Hilbert truncation}
\label{s5.1}
For the effective average action we make the following Ansatz, following the operational structure of \eqref{2.29}:
\begin{equation}\label{5.1}
	\begin{aligned}
\bar\Gamma_k[\hat Q, \bar Q]=& \frac{1}{G_{N,k}}\int d^Dx\left\{
	\frac{1}{4}[\det(Q)]^{\frac{1}{2(1+rD)}}[Q^{ac}\;Q^{bd}-u'\;Q^{ab}\;Q^{cd}]\dot{Q}_{ab}\dot{Q}_{cd}
	\right.\\
	&\qquad \qquad\qquad \qquad\left.-[\det(Q)]^{\frac{1}{2(1+rD)}}\;[R([\det(Q)]^{-\frac{r}{1+rD}} Q)-2\Lambda_k]\right\}\;. 
		\end{aligned}
\end{equation}
with $u' = \frac{1 + r (2 + D r)}{(1 + D r)^2}$ found in section \ref{s2.3}. We are going to refer to the first term as the kinetic term and the second term as the Ricci term with the cosmological constant.

Let us  start considering the case  $r = 0$. In order to derive the Hessian, namely the second functional derivative at fixed $\bar Q$, we exploit the background field method and expand
\begin{equation}
\bar\Gamma_k[H + \bar Q, \bar Q] = \bar\Gamma_k[\bar Q, \bar Q]+ O(H) + \bar\Gamma_k^{\text{quad}}[H, \bar Q]+ O(H^3)\;.
\end{equation}
 By taking two times the functional derivative wrt. $H$, the quadratic part gives the Hessian
 \begin{equation}
\bar \Gamma^{(2), abcd}[H, \bar Q]=\frac{{\delta \bar\Gamma}_k^{\text{quad}}[H,\bar{Q}]}{\delta H_{ab} \delta H_{cd}}\;.
 \end{equation}
  With the ansatz \eqref{5.1}, it takes the form:
\ba\label{5.2}
\bar \Gamma^{(2), abcd}[H, \bar Q]_{H=0}&=&
\frac{[\det(\bar Q)]^{1/2}}{G_{N,k}} \left((-\overline{\Delta}_{D+1}+2\Lambda_k) K^{abcd} + U_k^{abcd}\right)\\ \label{5.3}
K^{abcd}&=&\frac{1}{4}\bar Q^{ac}  \bar Q^{bd}+\frac{1}{4}\bar Q^{ad}  \bar Q^{bc}-\frac{1}{2}\bar Q^{ab} \bar Q^{cd}
\\
U_k^{abcd}&=&-\frac{1}{2}\left( \bar \nabla^{(a}\bar \nabla^{b)} \bar Q^{c d} + \bar \nabla^{(c}\bar \nabla^{d)} \bar Q^{a b} - \bar \nabla^{(a}\bar \nabla_i \bar Q^{i)c}\bar Q^{bd}-\bar \nabla^{(a}\bar \nabla_i \bar Q^{i)d}\bar Q^{bc}\right)\nonumber\\
&&-\frac{1}{2}(\bar R_{1}^{acbd}+\bar R_{1}^{adbc})-\frac{1}{4}\left(\bar Q^{ac} \bar R_1^{bd} + \bar Q^{ad}\bar R_1^{bc} +\bar Q^{bc} \bar R_1^{ad}+\bar Q^{db}  \bar R_1^{ac}\right)\label{5.4}\\
&&+\frac{1}{2}\left(\bar Q^{ab }\bar R_1^{cd}+ \bar Q^{cd} \bar R_1^{ab}\right)+\bar R_1\left(\bar Q^{ac}  \bar Q^{bd}+\bar Q^{ad}  \bar Q^{bc}-\bar Q^{ab} \bar Q^{cd}\right) \nonumber
\\&&+ \frac{\Lambda_k}{2}\left(3\bar Q^{a b}\bar Q ^{c d}-2(\bar Q^{ac}  \bar Q^{bd}+\bar Q^{ad}  \bar Q^{bc})\right)\nonumber
\ea
were the $\bar \nabla$ and the $\bar R_{\epsilon=1}\equiv\bar R_1$ are meant to be the $D+1$-dimensional  covariant derivative and 
the $D+1$-dimensional curvature invariants constructed on it, respectively, i.e.,  on the background with $N=1$ and $N^a = 0$. 
Since we are interested in the variation of the spatial components, these are evaluated only on the spatial components, 
explicitly that means that the indexes $a,b,c,d,i = 1,\dots, D$. The Laplacian $\overline \Delta_{D+1}$ is also restricted 
to the same class of backgrounds  introduced in \eqref{4.20}. Furthermore, we observe that the  $+2\Lambda_k$ term will act as a mass 
term in the propagator as follows from \eqref{5.2} which will become clearer later, when we need to evaluate the traces.

As discussed in the previous section, we wish now to use the covariant Laplacian $\overline \Delta_{D+1}$ and the associated heat kernel 
technology. In order to evaluate the r.h.s. of the flow equation and to compare with the l.h.s., we project the operators back into $S_D (A=0, B=2, w = 0)$. In this first investigation we don't take into account the corrections in the heat kernel arising from the projection between $S_D$ and $T_{D+1}$ (see appendix \ref{sa}). This is consistent with the choice of the regulator involving the projection of the standard heat kernel to $S_D$ as in \eqref{4.28}.

Let us illustrate with an example how to rearrange the terms, starting by the time derivative contribution, i.e., the first term in \eqref{5.1}. Performing the second functional derivative, one obtains:
\begin{equation}
\bar \Gamma_k^{(2), abcd}= \frac{1}{4 G_{N,k}} [\det(\bar Q)]^\frac{1}{2}[\bar Q^{a(c}\;\bar Q^{d)b}-\;\bar Q^{ab}\;\bar Q^{cd}](-\partial_t^2) + \text{ additional terms } = \frac{[\det(\bar Q)]^\frac{1}{2}}{G_{N,k}} K^{abcd} (-\partial_t^2) +  \cdots
\end{equation}
where we used $K$ in \eqref{5.2} and the additional terms not displayed are mixed contributions of the time derivative acting on the 
fluctuation and the background. All these terms are taken into account and are combined with the non-minimal operators and the curvature tensors.
Furthermore, combining with the second functional derivative coming from the Ricci scalar contribution, we can recast 
the contributions in terms the Laplacian $\overline \Delta_{D+1}$. Consider for example the operator to be applied to a scalar:
\begin{equation}
-	\bar g_{\mu \nu}\bar \nabla^\mu \bar \nabla^\nu \bigg|_{\bar g_{tt} = 1,\; \bar g_{ta} = 0,\; \bar g_{ab} = \bar Q_{ab}} =- \partial_t ^2- \bar Q_{ab}\bar \nabla^a \bar \nabla^b - \frac{1}{2}\bar Q^{ab} \partial_t \bar Q_{ab} \partial_t
\end{equation}
and again the non-minimal terms will get recombined with the last, non-minimal term in the first line of the potential-like term $U_k$ in \eqref{5.4}. Importantly, also the contribution coming from the Ricci term comes with the same index structure as the kinetic term, namely with $K^{abcd}$. This justifies the expression in  \eqref{5.2}.

In fact, this was anticipated based on the analysis presented in section 2, in particular in \eqref{2.35}, where we showed that modulo the Gibbons-Hawking boundary term, via the Codacci relation the foliated Lagrangian we are considering reproduces the $D+1$ Lagrangian on the foliated class of spacetimes with $\bar g_{tt} = 1$ and $\bar g_{ij} = Q_{ij}$. Hence, this will allow us to work with the heat kernel evaluated on $\overline \Delta_{D+1}$ and to then identify the beta functions at the ``unfoliated level".

\bigskip

Now let us consider a generic $r$ and $u'$ as in our Ansatz \eqref{5.1}. From the tensorial structure in front of the kinetic term in \eqref{5.1}
one might wonder, how the overall tensorial structure changes. This is indeed expected since we performed a canonical transformation as in 
\eqref{2.14} only on the spatial components, namely only on $q_{ab}$. We will show how the new tensorial structure of the kinetic term 
becomes $r$-dependent and will take the form of
\begin{equation}\label{5.7}
K_1^{abcd}(r) = \frac{1}{2}[Q^{a(c} Q^{d)b}- u' Q^{ab}Q^{cd}]
\end{equation}
as was anticipated in \eqref{2.29}.
 First of all, this can also be derived by evaluating $R_{\epsilon=1}$ for 
 $\bar g_{t\mu} = \delta^t_\mu$ and $\bar g_{ab} = [\det(\bar{Q})]^{-\frac{r}{1+Dr}}\bar{Q}_{ab}$. 
In particular, the terms which contain time derivatives  in $R_{\epsilon=1}$ are (dropping the bar)
\begin{eqnarray}
&&\frac{\sqrt{q}}{4}\left[\left(-3q^{ac} q^{bd}+q^{ab}q^{cd}\right)\dot q_{ab} \dot q_{cd}+ 4 q^{ab} \ddot{q}_{ab}\right] = \frac{\sqrt{q}}{4}\left(q^{ac}q^{bd}-q^{ab}q^{cd}\right)\dot q_{ab} \dot q_{cd}\nonumber\\
&=&\frac{\det(Q)^\frac{1}{2(1+rD)}}{4}\left[Q^{ac}Q^{bd}-Q^{ab}Q^{cd}\right]\left(\dot Q_{ab}-\frac{r}{1+rD}Q_{ab}Q^{ef}\dot Q_{ef}\right)\left(\dot Q_{cd}-\frac{r}{1+rD}Q_{cd}Q^{ef}\dot Q_{ef}\right)\nonumber\\
& =& \frac{\det(Q)^\frac{1}{2(1+rD)}}{4}\left[Q^{ac}Q^{bd}-u'Q^{ab}Q^{cd}\right]\dot{Q}_{ab} \dot{Q}_{cd}
\end{eqnarray}
Hence, we recovered the form of the Ansatz as in \eqref{5.1}, with $K^{abcd}(r)$ as the tensorial structure.

The canonical transformation, however, breaks the wished covariant tensorial symmetry with the spatial structure. When computing the Hessian, the spatial derivatives coming from the Ricci tensor contain a Laplacian plus non-minimal derivative terms as  in the first line of \eqref{5.4}. By expanding the spatial Laplacian we get the same structure as for the kinetic term
\begin{equation}
K^{abcd}(r=0)h_{ab}\bar{D}_i\bar{D}^i  h_{cd} =: K_1^{abcd}(r)H_{ab}  
\bar{D}_i\bar{D}^i H_{cd}
\end{equation}
where $i = 1,\dots,D$ and $\bar{D}_a \bar{q}_{bc}=0$. 
Note that the fluctuation $H$ of $Q$ can be expressed in terms of the fluctuation of $q$ as
\begin{eqnarray}
h_{ab} &=& 
[\det(\bar{Q})]^\frac{-r}{1+rD}\left[\left(H_{ab}-\bar{Q}_{ab} \frac{r}{1+rD}\bar{Q}^{cd} H_{cd}\right)+\left(H_{ab} \frac{r}{1+rD}\bar{Q}^{cd} H_{cd}+\frac{r (1 + r + D r)}{(1 +  rD)^2}\bar{Q}_{ab}  \bar{Q}^{c[e} \bar{Q}^{f]d} H_{ce} H_{fd}\right)\right]\nonumber\\
&=:&h^{(1)}_{ab}+h^{(2)}_{ab}\label{5.10}
\end{eqnarray}
A similar strategy is used to compute the entire Hessian \cite{rp}, given that the variations wrt. $q_{ab}$ have already been computed, we can exploit those contributions and express the Hessian wrt. $Q_{ab}$ through composition:
\begin{equation}
	\bar\Gamma^{(2)}[\hat Q, \bar Q] = \bar\Gamma^{(2)}[\hat q, \bar q]^{ijkl}\frac{\delta h_{ij}}{\delta H_{ab}} \frac{\delta h_{kl}}{\delta H_{cd}}+2 \bar\Gamma^{(1)}[\hat q, \bar q]^{ij}\frac{\delta^2 h_{ij}}{\delta H_{ab}\delta H_{cd}} 
\end{equation}
In effect, we are just interested in the terms quadratic in $H$, allowing us to consider just $h_{ab}^{(1)}$ for the first term and
$h_{ab}^{(2)}$ for the second term, defined in \eqref{5.10}.  This gives
\ba\label{5.12}
\bar{\Gamma}^{(2),abcd}[H,\bar{Q}]^{abcd}_{H=0}&=&
\frac{[\det(\bar Q)]^{\frac{1}{2+2rD}}}{G_{N,k}} \left(-K_1^{abcd} \partial_t^2
+(-\overline{\Delta}'_{D}+2\Lambda_k) K_2^{abcd} + U_k^{abcd}\right)\\
K_1^{abcd}&=&\frac{1}{4}\bar Q^{ac}  \bar Q^{bd}+\frac{1}{4}\bar Q^{ad}  \bar Q^{bc}-\frac{u'}{2}\bar Q^{ab} \bar Q^{cd}
\\
K_2^{abcd}&=&\frac{1}{4}\bar Q^{ac}  \bar Q^{bd}+\frac{1}{4}\bar Q^{ad}  \bar Q^{bc}-\frac{u''}{2}\bar Q^{ab} \bar Q^{cd}
\\
U_k^{abcd}&=&-\frac{1}{2}\left(\frac{1 + 2 r}{1+Dr}( \bar \nabla^{(a}\bar \nabla^{b)} \bar Q^{c d} + \bar \nabla^{(c}\bar \nabla^{d)} \bar Q^{a b} - \bar \nabla^{(a}\bar \nabla_i \bar Q^{i)c}\bar Q^{bd}-\bar \nabla^{(a}\bar \nabla_i \bar Q^{i)d}\bar Q^{bc}\right)\nonumber\\
&&-\frac{1}{2}(\bar R_1^{acbd}+\bar R_1^{adbc}-\frac{1}{4}\left(\bar Q^{ac} \bar R_1^{bd} + \bar Q^{ad}\bar R_1^{bc} +\bar Q^{bc} \bar R_1^{ad}+\bar Q^{db}  \bar R_1^{ac}\right)\label{5.15}\\
&&+\frac{1}{2}\frac{1+2r}{1+rD}\left(\bar Q^{ab} \bar R_1^{cd}+ \bar Q^{cd} \bar R_1^{ab}\right)+\frac{1}{4}\bar R_1\left(\frac{1+2r}{1+Dr}(\bar Q^{ac}  \bar Q^{bd}+\bar Q^{ad}  \bar Q^{bc})-(1+2r)^2\bar Q^{a b} \bar Q^{c d}\right)
\nonumber\\
&& + \frac{1}{2} \Lambda_k \left(\frac{-2 (2 + 3 D r + D^2 r^2)}{(1+Dr)^2}(\bar Q^{ac}  \bar Q^{bd}+\bar Q^{ad}  \bar Q^{bc})+\frac{3 + 6 r + 2 (1 + D) r^2}{(1+Dr)^2}\bar Q^{a b}\bar Q ^{c d}\right)\nonumber
\ea
where the indexes are only spatial, i.e., $a,b, c,d,i = 0,\dots D$ and $\overline{\Delta}'_D$ 
is the Laplacian of $\bar{q}_{ab}$ which is not quite the same as 
$\overline{\Delta}_D-\partial_t^2$ but 
the difference is due to $[\overline{\Delta}_{D+1}, E\cdot E^\ast]\not=0$ which we will ignore 
in what follows in accordance with what we said at the end of section \ref{s4}. 
Note the difference between the tensors in \eqref{5.12} in front of temporal-temporal and 
spatial-spatial derivatives.
This is due to the non-minimal contributions of the Hessian evaluated on the expansion 
\eqref{5.10} which lead to additional terms in the Laplacian of the trace of the fluctuation 
$H_{ab}$, resulting effectively in
\begin{equation}
	K_2^{abcd}(r) = \frac{1}{2}[Q^{a(c} Q^{d)b}- u'' Q^{ab}Q^{cd}], \qquad u'' = \frac{1 + r (4 + (2 + D) r)}{(1+D r)^2}
\end{equation}
Crucially, this is the manifestation of the fact, that the canonical transformation has broken the ``foliated covariance", 
which was instead realized above, for the special case $r =0$. In general, the kinetic term and 
the Ricci term flow differently.
This is physically correct as we take here the effect of non-trivial Jacobean that arises for $r=0$ 
and which usually discarded by hand seriously into account.

In order to compute the flow one either has to go beyond the Einstein-Hilbert truncation and expand the theory space 
or we must take an additional step that is not necessary for $r=0$ which we interpret as part of taking the Einstein-Hilbert
truncation. For the purpose of the present paper, we will choose the latter option, reserving the former more fundamental 
option for future research. The additional truncation step consists in taking a ``covariant average" between the tensorial structures 
of $K_1$ and $K_2$ by constructing
\begin{equation}\label{5.13}
K_{\pm}^{abcd} = \frac{1}{2}[Q^{a(c} Q^{d)b}- u_{\pm} Q^{ab}Q^{cd}], \qquad u_{\pm} = \frac{u'\pm u''}{2}
\end{equation}
We stick to the minimal set of operators dictated by the Einstein--Hilbert truncation, setting $K_-^{abcd} = 0$ and using $K_+^{abcd} $, with 
\begin{equation}
u_+ = \frac{1 + r (3 + r + D r)}{(1 + D r)^2}
\end{equation}
as the prefactor of the trace part in the $D+1$ tensorial structure. 

As a final consistency check, we note that setting $r = 0$, the previous results are found. Having now all the ingredients at hand, we can now evaluate the flow equation.

\bigskip
In particular, we will go up to first order in the expansion of the Wetterich identity \eqref{4.32}
\be \label{exp}
k\partial_k \Gamma_k=\frac{1}{2} 
\text{Tr}
\{
[k\partial_k R_k]\;P_k^{-1}
(1-[(R_k+U_k)\; P_k^{-1}])
\}
\ee
where we split the Hessian in a propagator $P_k$ and a potential term $U_k$:
\be
\bar{\Gamma}_k^{(2)}(H,\bar{Q})=:P_k+U_k\;
\ee
The propagator reads
\begin{equation}
P_k^{abcd}=G_{N,k}^{-1}K_+^{abcd}(-\overline{\Delta}_{D+1}+2\Lambda_k)\;,
\end{equation}
while the potential like term can be read off from \eqref{5.4}.

Making the Ansatz that the regulator $R_k$ has the tensorial structure 
\be\label{5.22}
R_k ^{abcd}{}= [\det(\bar Q)]^{\frac{1}{2(1+rD)}} G_{N,k}^{-1} K_+^{a bc d} 
E^\ast \cdot R_k(\overline{\Delta}_{D+1}) \cdot E,\; 
R_k(z)=k^2\; r(z/k^2)
\ee
simplifies considerably the computations. Hence, via 
$K_+$ the regulator will be adapted to the specific canonical transformation. 
Furthermore, as a regulator kernel, we choose the form introduced in \eqref{4.25}.
Finally, consistent with the approximation made at the end of section \ref{s4}, 
as a first step, we 
will ignore corrections coming from $[\overline{\Delta}_{D+1},E\cdot E^\ast]\not=0$. 

Effectively, one is left with computing the suitable trace of products of $K_+^{-1} U_k$ up to the order established in the expansion \eqref{exp}, where the matrix $K_+^{-1}$  is the inverse of $K_+$ defined in \eqref{5.13}:
\be
(K^{-1})^{abcd}_+ = \bar Q^{ac}  \bar Q^{bd}+\bar Q^{ad}  \bar Q^{bc}- \frac{(1 + 3 r + (1 + D) r^2) }{  D (1 + r + r^2)-1} \bar Q^{ab} \bar Q ^{cd}\;.
\ee

 \subsection{Heat kernel traces}
In this section, we will explicitly express the traces in \eqref{exp} as heat kernel traces. As is standard in RG analysis, we switch to dimensionless variables:
  \be
 y=z/k^2, \qquad
 \Lambda_k = \lambda_kk^2, \qquad 
 G_{N,k} = \frac{g_k}{k^{D-1}}, \qquad
 U_k = u_k k^2, \qquad
 \eta_N = \frac{k \partial_k g_k}{g_k}\;,
 \ee
 where $\eta_N$ represents  the anomalous dimension of the dimensionless Newton's coupling and $z$ is the eigenvalue of $-\overline{\Delta}_{D+1}$.
 
As a next step, we need to evaluate the heat kernel traces. 
In order to do so, we will work at the level of the proper time  
(or heat kernel time) integration, considering the convolution of up to four proper time variables to obtain the expansion order required.\footnote{We refer the reader to \cite{ab} for a   proper time treatment of Einstein--Hilbert gravity by means of a one-loop improved flow equation.} Further information about the methods used to compute those proper time integrals can be found in appendix \ref{sb}, where we detail the procedure making use of the Barnes identity.

The cutoff function \eqref{5.22} has a suitable kernel in proper time given by
\begin{equation}
r(z/k^2) = \int_0^\infty \; ds \;e^{-s^2-s^{-2}}\;e^{-s\frac{z}{k^2}}
\end{equation}
where $z$ is the eigenvalue of $-\overline{\Delta}_{D+1}$. Notice that due to the form of the cutoff, when we take the $k$-derivative in the numerator of the flow equation, we also have to evaluate terms with the derivative of the cutoff kernel, which adds a contribution equal to:
\begin{equation}
\int ds	\;k \frac{d}{dk} \left(e^{-\frac{s z}{k^2}}\right)e^{-s^2-s^{-2}} = -2\int ds	\;s\frac{d}{ds} \left(e^{-\frac{s z}{k^2}}\right)e^{-s^2-s^{-2}} =  2 \int ds	\;e^{-\frac{s z}{k^2}}\left(e^{-s^2-s^{-2}}+s \frac{d}{ds}e^{-s^2-s^{-2}}\right)
\end{equation}
Note that due to the choice of cut-off function which is of rapid decrease at both zero and infinity, no boundary terms arise.

Based on what we need, let us now briefly review the heat kernel technology based on the Schwinger heat kernel time representation, which allows us to find a representation for the propagator
\begin{equation}
\frac{	1}{y+2\lambda_k} = \int_0^\infty e^{-(y+2\lambda_k)s}\; ds\;,
\end{equation}
where $y = -\overline{\Delta}_{D+1}/k^2$.
This representation allows us to exploit the heat kernel trace:
\begin{equation}
	\text{Tr} \; \left[\frac{	1}{y+2\lambda_k}\right] = \int_0^\infty \text {Tr}\left[e^{-ys}\right]\; e^{-2\lambda_ks}\; ds\;,
\end{equation}
since the term $\text{Tr}[e^{-ys}] = \text{Tr}[e^{\bar \Delta_{D+1}s}] $ can be recognized to be the heat kernel trace. In addition to this, we also have to consider the cutoff function, which also has a heat kernel time representation and can be included in the $s$-integration.
As an example, consider the first term on the l.h.s. of \eqref{exp}. The heat kernel representation of this term  amounts to evaluate
\begin{eqnarray}\label{5.31}
&\text{Tr}&\; \left[\frac{k \partial_k R_k }{P_k}\right]=\text{Tr} \;\left[\frac{((D-1)-\eta_N+2)r(y)+k\frac{d}{dk}r(y)}{y+2\lambda_k}\right] \nonumber\\&=& \int _0^\infty ds_1 \int _0^\infty ds_2\; \text{Tr}\; \left[e^{-y(s_1+s_2)}\right]e^{-2\lambda_ks_1}\left(((D-1)-\eta_N+2)e^{-s_2^2-s_2^{-2}}+2s_2 \frac{d}{ds_2}e^{-s_2^2-s_2^{-2}}\right)\;,
\end{eqnarray}
where  $\eta_N$ appears because of the $k$-th derivative of $G_{N,k}^{-1}$ in the regulator \eqref{5.22}.  We now  exploit the heat kernel expansion for the trace at the truncation desired, i.e.,
\begin{equation}\label{5.30}
 \text{tr}\;\left[e^{\bar \Delta_{D+1}s}\right] = \frac{1}{(4\pi s)^{\frac{D+1}{2}}}\left(1+\frac{\bar R_1 s}{6}+ \cdots\right)\;.
\end{equation}
We specify here, that the above trace tr$[\;\cdot\;]$ acts only on the internal space as 
$\int d^{D+1}x \sqrt{g} \langle x |\; \cdot \;|x\rangle$, and hence does not include the index contraction over the field space which is 
performed in Tr$[\;\cdot\;]$.
This means that in the first term in \eqref{5.31} two convoluted heat kernel time integrations have to be performed
\begin{equation}
\frac{1}{(4\pi)^{\frac{D+1}{2}}} \int _0^\infty ds_1 \int _0^\infty ds_2\; \frac{e^{-2\lambda_ks_1}e^{-s_2^2-s_2^{-2}}}{(s_1+s_2)^{\frac{D+1}{2}}}\;\left(1+\frac{\bar R_1 (s_1+s_2)}{6}+ \cdots\right)\;.
\end{equation}
The techniques to solve these convoluted integrals and find an analytic expression as an expansion in $\lambda_k$ can be found in appendix \ref{sb}.

As it happens, the second order term \eqref{exp} will contain up to three convoluted heat kernel time integrals, respectively. Exploiting that
\begin{equation}
	\frac{	1}{(y+2\lambda_k)^2} = \int_0^\infty s\;e^{-(y+2\lambda_k)s}\; ds\;,
\end{equation}
the term with the potential amounts to:
\begin{eqnarray}
\text{Tr}\; \left[\frac{k \partial_k R_k U_k}{P_k^2}\right]&=&\text{Tr} \;\left[\frac{\left(((D-1)-\eta_N+2)r(y)+2 s\frac{d}{ds}r(y)\right)u_k K^{-1}_+}{(y+2\lambda_k)^2}\right]\nonumber\\
& =& \int _0^\infty ds_1 \int _0^\infty ds_2\; \text{Tr}\; \left[e^{-y(s_1+s_2)}K_+^{-1}u_k\right]s_1e^{-2\lambda_ks_1}
\nonumber
\\ &&\qquad\qquad\times\left(((D-1)-\eta_N+2)e^{-s_2^2-s_2^{-2}}+2s_2 \frac{d}{ds_2}e^{-s_2^2-s_2^{-2}}\right)\;
\end{eqnarray}
In order to evaluate this trace, we will exploit both the expansion in \eqref{5.30} for the minimal terms in the potential (second and third line in \eqref{5.4} and second to fourth line in \eqref{5.15}), and the expansion for the non-minimal terms (first line in the potential), given by:
\begin{equation}
\text{tr}\; \left[\bar \nabla_{(\mu}	\bar \nabla_{\nu)}e^{\bar \Delta_{D+1}s}\right] = \frac{1}{( 4\pi s)^\frac{ D+1}{2}}\left(-\frac{1}{2s}\bar g_{\mu \nu}
-\frac{1}{2}\bar g_{\mu \nu} \frac{\bar R_1}{6}+\bar \nabla_{(\mu}\bar \nabla_{\nu)}\frac{\bar R_{1,\mu \nu}}{6} \right)\,.
\end{equation}
In particular then, performing  the complete trace $\text{Tr}[\;\cdot\;]$ also 
over the indexes as in the flow equation, together with the potential $U_k$ and $K^{-1}_+$ we obtain:
\begin{eqnarray}
\text{Tr}\; \left[e^{-ys}K_+^{-1}U_k\right]^\text{min}&=&\frac{1}{(4\pi s)^{\frac{D+1}{2}}}\Bigg[-\left(D(1+D)+ \frac{(2+D)(D-1)}{1+Dr}+ \frac{D-2-2Dr}{D(1+r+r^2)-1}\right)\left(\Lambda_k + \frac{\bar R_1}{6}s\right)\nonumber\\
\nonumber
&&\qquad\qquad\quad- \frac{(D-2) (3 + D + 4 r (2 + r) + 2 D r^2 (2 + r) + 
	D^2 ( r^3-1))}{2 (1 + D r) (-1 + D (1 + r + r^2))} \Lambda_k\bar R_1s \Bigg] 
\\&=:&\frac{1}{(4\pi s)^{\frac{D+1}{2}}}\left[a_1 \Lambda_k + a_2 \bar R_1 s+a_3 \Lambda_k \bar R_1 s\right]
\label{5.35}
\end{eqnarray}
\begin{eqnarray} 
\text{Tr}\; \left[e^{-ys}K_+^{-1}U_k\right]^\text{non-min}&=&\frac{1}{(4\pi s)^{\frac{D+1}{2}}}\frac{1 + 2 r (1 + r) - 3 D r (1 + r) - D^2 (1 + r + r^2)}{3 (  D (1 + r + r^2)-1)} \left(\frac{D}{s}+\frac{D-2}{6}\bar R_1 \right)\nonumber \\
&=:& \frac{1}{(4\pi s)^{\frac{D+1}{2}}}\left[ \frac{a_4}{s}+a_5 \bar R_1 \right]\label{5.36}
\end{eqnarray}
where we have encoded in  $a_1, \dots, a_5$ the $r-$ and $D-$ dependent coefficients.

Finally, the contribution with the regulator in the second order term is given by
\begin{eqnarray}
	\text{Tr}\; \left[\frac{(k \partial_k R_k) \;R_k}{P_k^2}\right]&=&\text{Tr} \;\left[\frac{\left(((D-1)-\eta_N+2)r(y)+2 s\frac{d}{ds}r(y)\right)r(y)}{(y+2\lambda_k)^2}\right]\nonumber\\
	& =& \int _0^\infty ds_1 \int _0^\infty ds_2\int_0^\infty ds_3\; \text{Tr}\; \left[e^{-y(s_1+s_2)}K_+^{-1}u_k\right]s_1e^{-2\lambda_ks_1}e^{-s_3^2-s_3^{-2}}\\
	\nonumber
&&\qquad \qquad\times	\left(((D-1)-\eta_N+2)e^{-s_2^2-s_2^{-2}}+2s_2 \frac{d}{ds_2}e^{-s_2^2-s_2^{-2}}\right)\;
\end{eqnarray}
Even if the evaluation of this integral is a bit more involved, it can be expressed as an analytic expansion in $\lambda_k$. The reader can find the details in appendix \ref{sb}.

In order to facilitate the notation, we will introduce the following symbolic expressions for the proper time integrals. Let us denote by 
\begin{equation}
I_{2,p}(\lambda_k):= \int_{0}^\infty ds_1  \int_{0}^\infty ds_2\; \frac{e^{-2\lambda_ks_1} e^{-s_2-s^{-2}_2}}{(s_1 + s_2)^p}
\end{equation}
\begin{equation}
	J_{2,p}(\lambda_k):= \int_{0}^\infty ds_1  \int_{0}^\infty ds_2\; \frac{e^{-2\lambda_ks_1}s_2 \frac{d}{ds_2}e^{-s_2-s^{-2}_2}}{(s_1 + s_2)^p}
\end{equation}
\begin{equation}
	Y_{2,p}(\lambda_k):= \int_{0}^\infty ds_1  \int_{0}^\infty ds_2\; \frac{s_1e^{-2\lambda_ks_1} e^{-s_2-s^{-2}_2}}{(s_1 + s_2)^p}
\end{equation}
\begin{equation}
	K_{2,p}(\lambda_k):= \int_{0}^\infty ds_1  \int_{0}^\infty ds_2\; \frac{s_1e^{-2\lambda_ks_1}s_2 \frac{d}{ds_2}e^{-s_2-s^{-2}_2}}{(s_1 + s_2)^p}
\end{equation}
\begin{equation}
I_{3,p}(\lambda_k):= \int_{0}^\infty ds_1  \int_{0}^\infty ds_2 \int_{0}^\infty ds_3\; \frac{e^{-2\lambda_ks_1} e^{-s_2-s^{-2}_2}e^{-s_3-s^{-2}_3}}{(s_1 + s_2+s_3)^p}
\end{equation}
\begin{equation}
	J_{3,p}(\lambda_k):= \int_{0}^\infty ds_1  \int_{0}^\infty ds_2 \int_{0}^\infty ds_3\; \frac{e^{-2\lambda_ks_1} s_2 \frac{d}{ds_2}e^{-s_2-s^{-2}_2}e^{-s_3-s^{-2}_3}}{(s_1 + s_2+s_3)^p}
\end{equation}
 the integrals needed in our calculations. The index $p$ will be determined by the heat kernel expansion, while the first index denotes the number of proper time integrations to be performed.

This completes the analysis of the treatment of the traces via the heat kernel methods.

\subsection{Beta functions}
\label{s5.2}
After having evaluated  the traces, we can now return to \eqref{exp} and compare   the l.h.s with the r.h.s. of the flow. In particular, in the Einstein--Hilbert truncation we are left with the flow of the two dimensionless gravitational coupling constants. These can be identified by matching terms on both sides of the equation. Those proportional to the identity operator, which yield $k \partial_k \left(\lambda_k/g_k\right)$, and those proportional to the Ricci scalar, which provide 
$k \partial_k \lambda_k$. One can then carefully disentangle them to find $k \partial_k g_k$. 
For the sake of readability, we will report the beta functions in terms of the $a_1, \dots a_5$, the $r$- and $D$-dependent coefficients, which can be read off from \eqref{5.35} and \eqref{5.36}. The resulting flow equations for the two dimensionless coupling constants are:
\ba\label{5.44}
k \partial_k \lambda_k&=& -(D+1) \lambda_k + \eta_N \lambda_k + \frac{g}{(4 \pi)^{\frac{D+1}{2}-1} }\Bigg\{(D-1)-\eta_N +2)\bigg( \frac{D (D+1)}{2}I_{2,2}(\lambda_k)\\
\nonumber
&&-(a_1Y_{2,2}(\lambda_k)\lambda_k+a_4Y_{2,3}(\lambda_k))-\frac{D (D+1)}{2}I_{3,2}(\lambda_k) \Bigg)\\
\nonumber
&&+2\bigg( \frac{D (D+1)}{2}J_{2,2}(\lambda_k)-(a_1K_{2,2}(\lambda_k)\lambda_k+a_4K_{2,3}(\lambda_k))-\frac{D (D+1)}{2}J_{3,2}(\lambda_k)\Bigg) \Bigg\}
\\
k \partial_k g_k&=& (D-1) g + \eta_N- \frac{2g^2}{(4 \pi)^{\frac{D+1}{2}-1} }\Bigg\{(D-1)-\eta_N +2)\bigg( \frac{D (D+1)}{12}I_{2,1}(\lambda_k)\label{5.45}\\
\nonumber
&&-(a_2Y_{2,1}(\lambda_k)+a_3Y_{2,1}(\lambda_k)\lambda+a_5 Y_{2,2}(\lambda_k))-\frac{D (D+1)}{12}I_{3,1}(\lambda_k) \Bigg)\\
\nonumber
&&+2\bigg( \frac{D (D+1)}{12}J_{2,1}(\lambda_k)-(a_2K_{2,1}(\lambda_k)\lambda_k+a_3K_{2,1}(\lambda_k)+a_5 K_{2,2}(\lambda_k))-\frac{D (D+1)}{12}J_{3,1}(\lambda_k)\Bigg) \Bigg\}\;.
\ea
We will now evaluate the integrals via the series expansion in appendix \ref{sb} in order to analyze the beta function.

\subsection{Fixed points and critical exponents of dimensionfree couplings}
\label{s5.3}
Once the beta functions have been determined, we can look for the fixed points of the theory, i.e., whether the beta functions vanish in the 
limits $k \to 0$ and $k \to \infty$, and determine the values that the coupling constants take at those fixed points.

First of all, we specialize to the case where $r= 0$ and $D = 3$, where no canonical transformation has been performed
and the effect of the Jacobean on the flow is incorrectly abandoned by hand. That is, one incorrectly  does not 
take into account the non trivial Jacobian, which arises for the choice $r=0$ of the Weyl algebra.
The system exhibits an IR Gaussian fixed point at $\lambda=g=0$ with critical exponents equal to the canonical mass dimensions:
\begin{equation}
\theta_{1}^{UV} = 2\;, \qquad \theta_{2}^{UV} = -2\;.
\end{equation}
This agrees with the standard ASQG Einstein--Hilbert truncation \cite{11b}.
Furthermore, evaluating the beta functions  \eqref{5.44} and \eqref{5.45} we find the UV-fixed point at
\begin{equation}
 \lambda_* = 1.85 \;, \qquad \quad g_* = 58.73\;,
\end{equation}
the analogue of the Reuter fixed point \cite{11b}.
In Figure \ref{fig:1} we plot the $\lambda-g$-phase diagram, where the trajectories and the fixed points are depicted.
\begin{figure}[H]
	\centering
	\includegraphics[width=.57\textwidth]{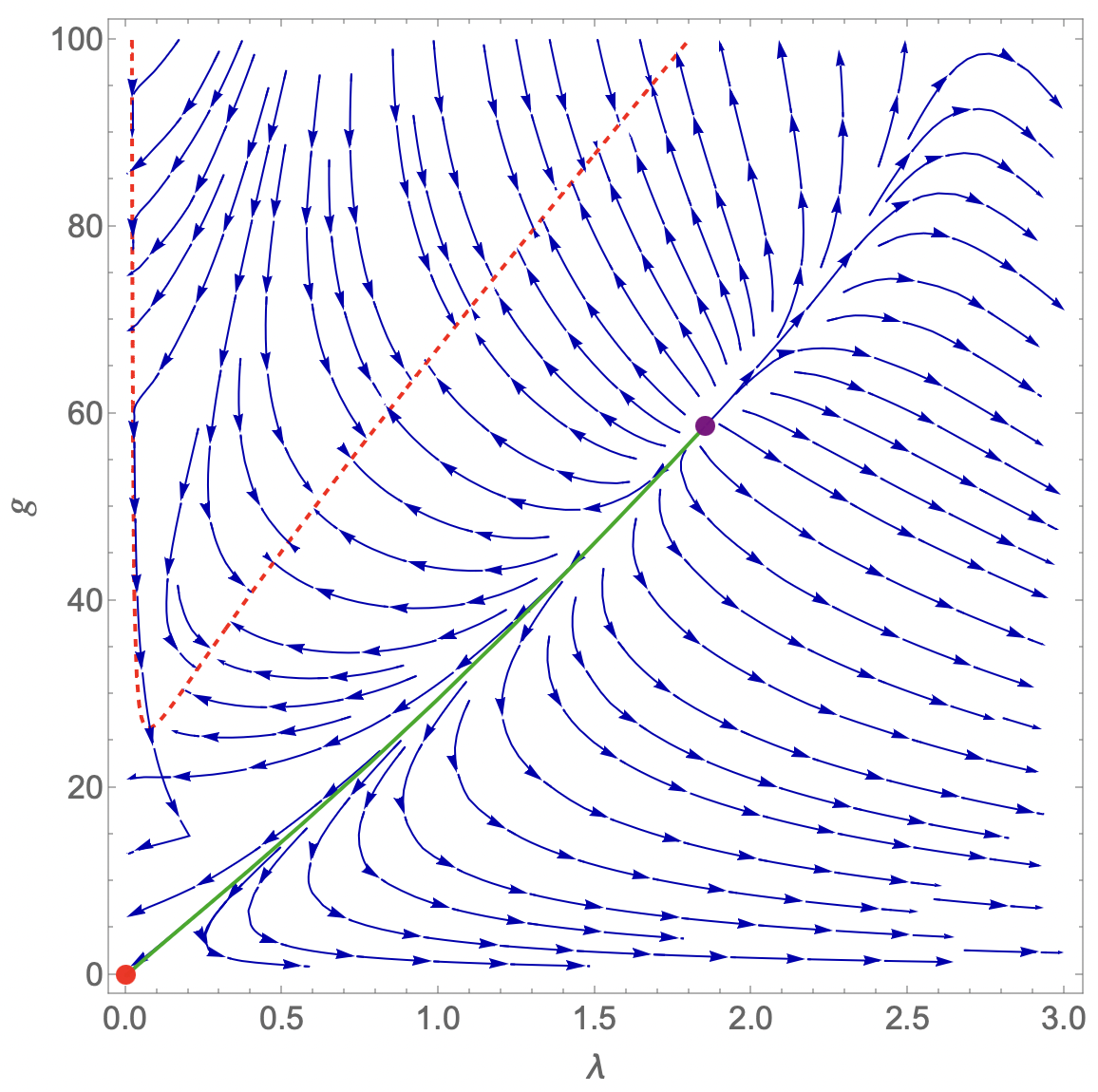}
	\caption{Flow diagram in the $\lambda-g$ plane for $r = 0$ and $D=3$. The arrows point towards decreasing values of $k$. All trajectories stem from the fixed point (purple point). The red dashed line is the ``curtain" or the singular locus: beyond that, the flow cannot be trusted anymore. Technically it represents the  parametric curve in the denominator of the beta functions. The flow cannot be followed anymore below the values at which the denominator vanishes and the beta functions become singular. The green line represents the separatrix (trajectory of the Type IIa), connecting the UV fixed point with the IR Gaussian fixed point, depicted by a red point.}
	\label{fig:1}
\end{figure}

Furthermore we compute the critical exponent, which determines how the coupling constants scale around the fixed point.  We find
\begin{equation}
\theta^{UV}_{1} = 8.01\;, \qquad \quad \theta^{UV}_2 = 2.03\;.
\end{equation}
These are both real and positive, signaling the fact that the coupling constants are related to two relevant directions.

It would be interesting to compare our results with those recently obtained within the foliated fluctuation approach in ASQG \cite{3b}. 
We notice that our critical exponents are real, as a subset of critical exponents found in \cite{3b}. The value of the coupling 
constants at the UV fixed point, however, differs significantly. We highlight that qualitative and technical differences between 
our approach and \cite{3b} exist, which may necessitate careful interpretation of a  comparison.

\bigskip
Let us now draw our attention to the special case $r = \frac{D-4}{4D}$ found in \eqref{2.23} and $D=3$. The IR fixed point persists 
also with this modified Jacobian. Regarding the UV fixed point in this case it takes the value
\begin{equation}
	 \lambda_* = 1.92 \;, \qquad \quad g_* = 57.41\;.
\end{equation}
The corresponding critical exponents are
\begin{equation}
	\theta^{UV}_{1} = 8.01\;, \qquad \quad \theta^{UV}_2 = 2.13\;.
\end{equation}
Comparing with the case $r = 0$ and $D=3$, one can notice that there is a minimal difference, which  does not affect the qualitative 
behavior of the flow diagram and the critical properties of the system. The fixed point is qualitatively similar as the one found 
previously and only one critical exponent is slightly modified. Thus, we can conclude that the RG properties of this system are 
minimally affected by using the correct choice of $r$ which avoids the Jacobian, 
which we accounted for by performing a canonical transformation, at least when we stay in the Einstein-Hilbert truncation.
Note, however, that this might change when we expand the theory space as outlined in the previous section.

\begin{figure}[h]
	\centering
	\includegraphics[width=.57\textwidth]{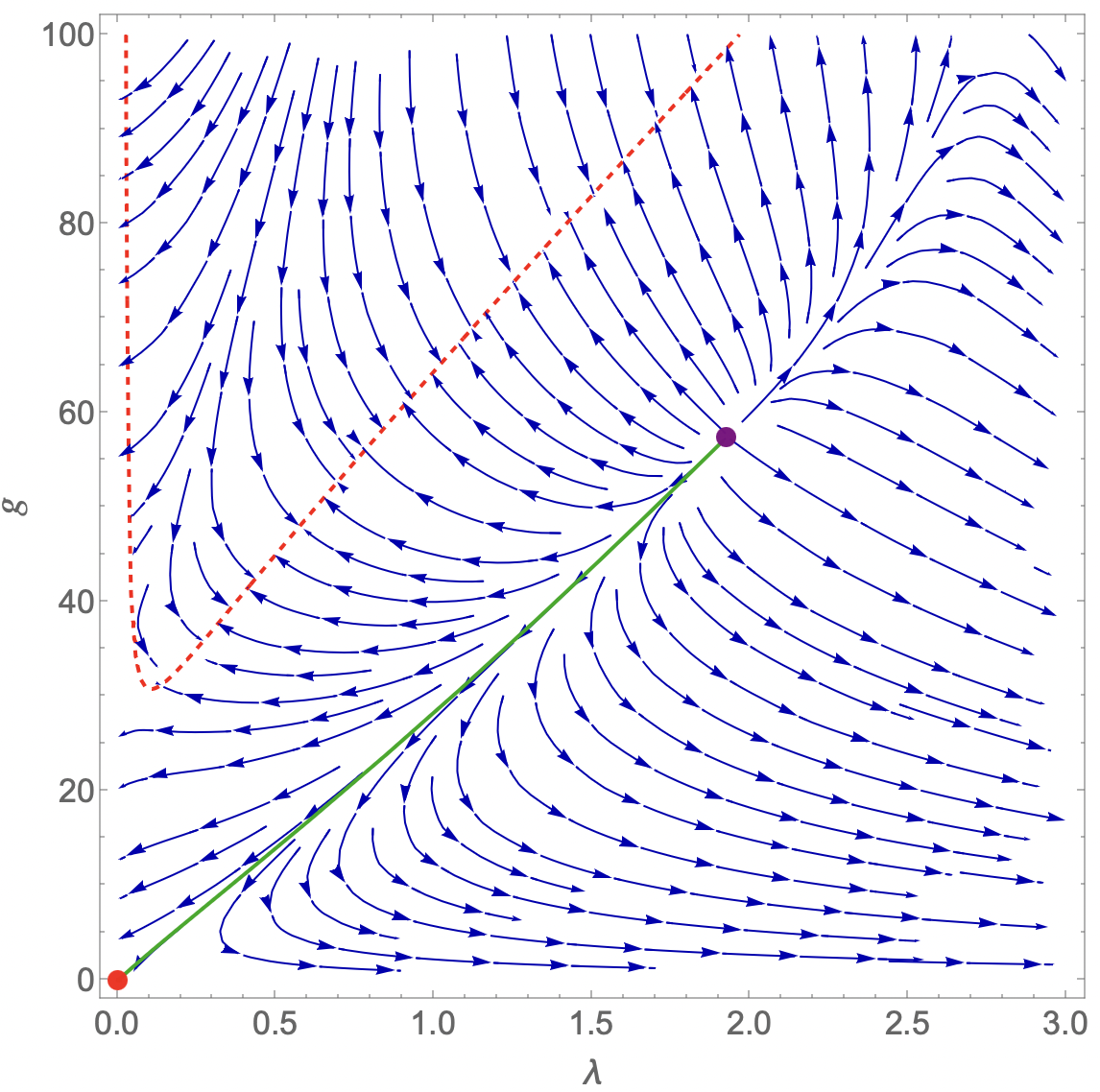}
	\caption{The flow diagram in the $\lambda-g$ plane illustrates trajectories with 
arrows indicating decreasing values of $k$ for $r=\frac{D-4}{4D}$, $D=3$. 
All trajectories originate from the fixed point (marked as a purple point). 
The red dashed line  marks the boundary beyond which the flow is no longer reliable.  
The green line denotes the separatrix, which connects the UV fixed point to the IR 
Gaussian fixed point (red point).
	}
	\label{fig:4}
\end{figure}

\subsection{$k=0$ limits of dimensionful couplings}
\label{s5.4}
Having found the flow, we can now integrate down the beta functions to $k \to 0$. This can be achieved for the separatrix and for the trajectories which flow towards increasing positive values of $\lambda$. In Figure \ref{fig:2} and \ref{fig:3} we report the dependence on $k$ of the dimensionless and the dimensionful cosmological constant and Newton's constant, respectively. As a trajectory, we picked the separatrix, namely that trajectory which flows from the UV fixed point into the IR Gaussian fixed point.
\begin{figure}[H]
	\centering
	\includegraphics[width=.48\textwidth]{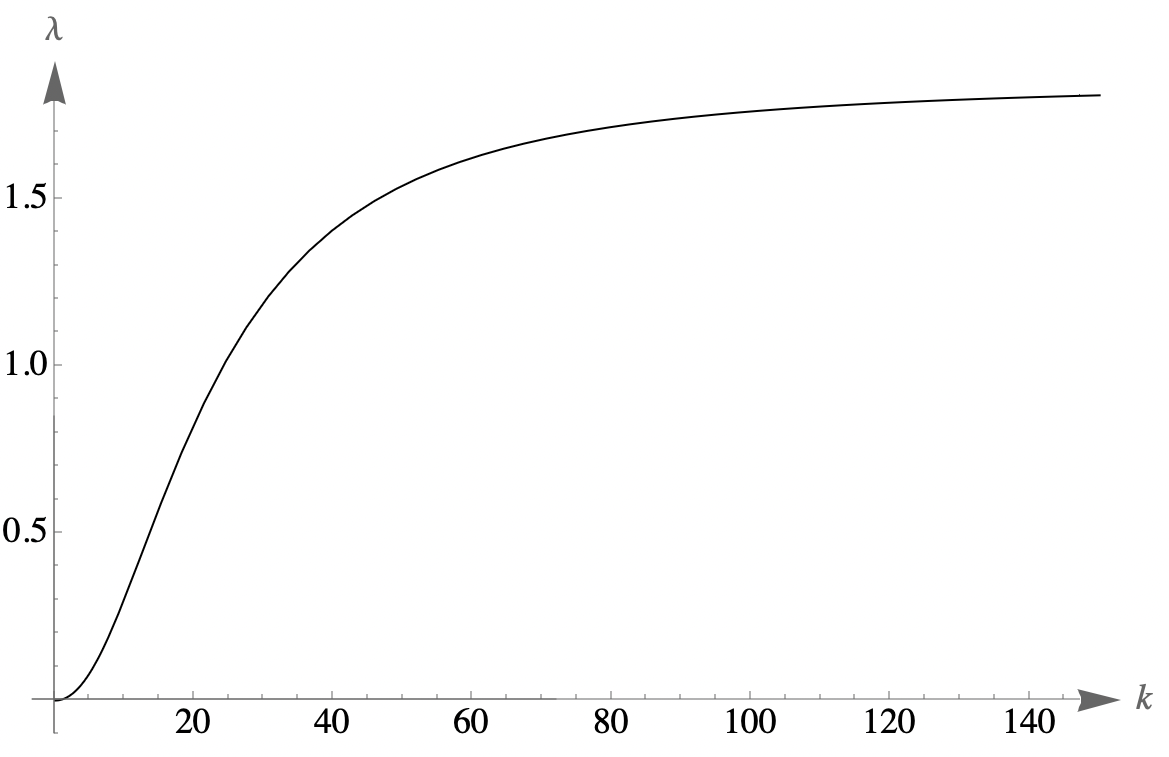}\quad
	\includegraphics[width=.48\textwidth]{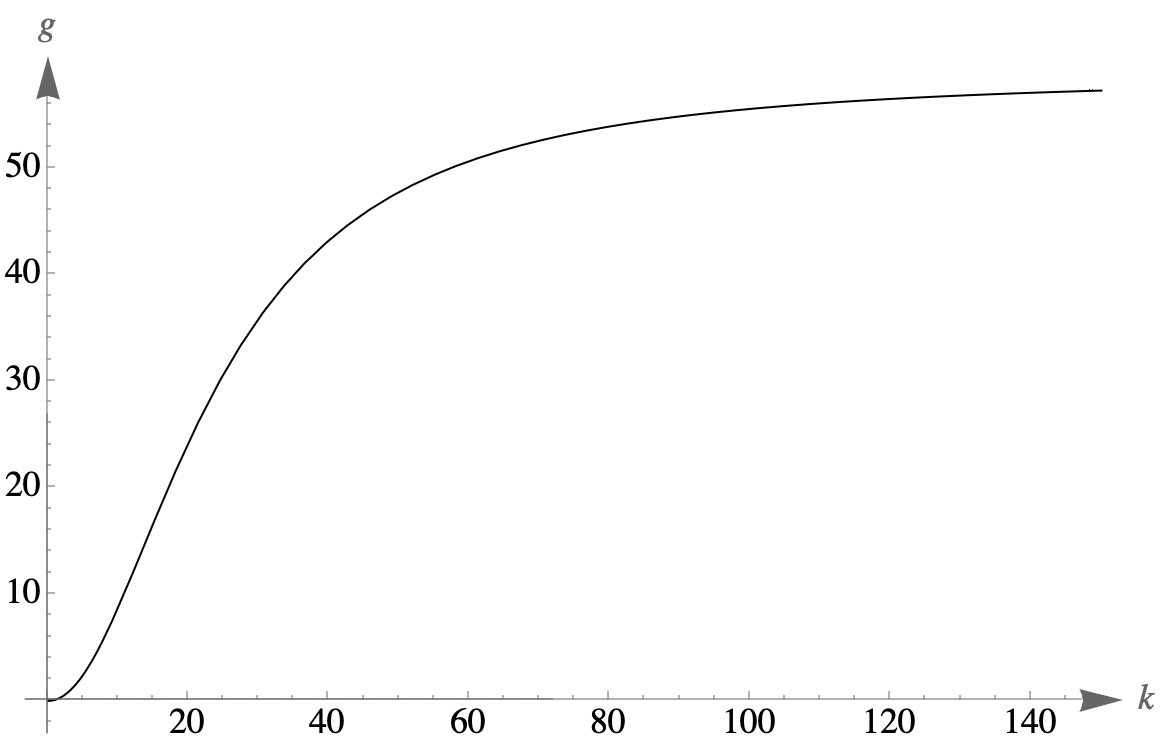}
	\caption{Plot of the $k$-dependence of the dimensionless coupling constants when $r = 0$ and $D = 3$. The trajectory chosen is the separatrix and the couplings go from 0 when $k=0$ to their UV fixed point when $k \to \infty$.}
	\label{fig:2}
\end{figure}

\begin{figure}[H]
	\centering
	\includegraphics[width=.48\textwidth]{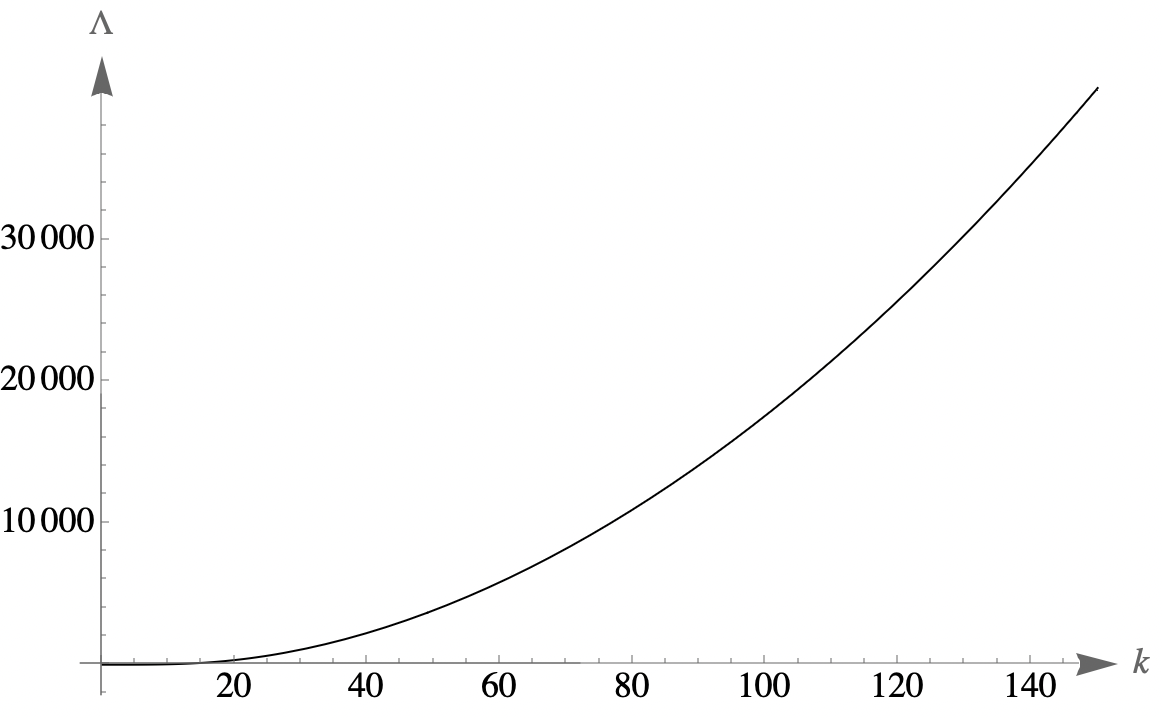}\quad
		\includegraphics[width=.48\textwidth]{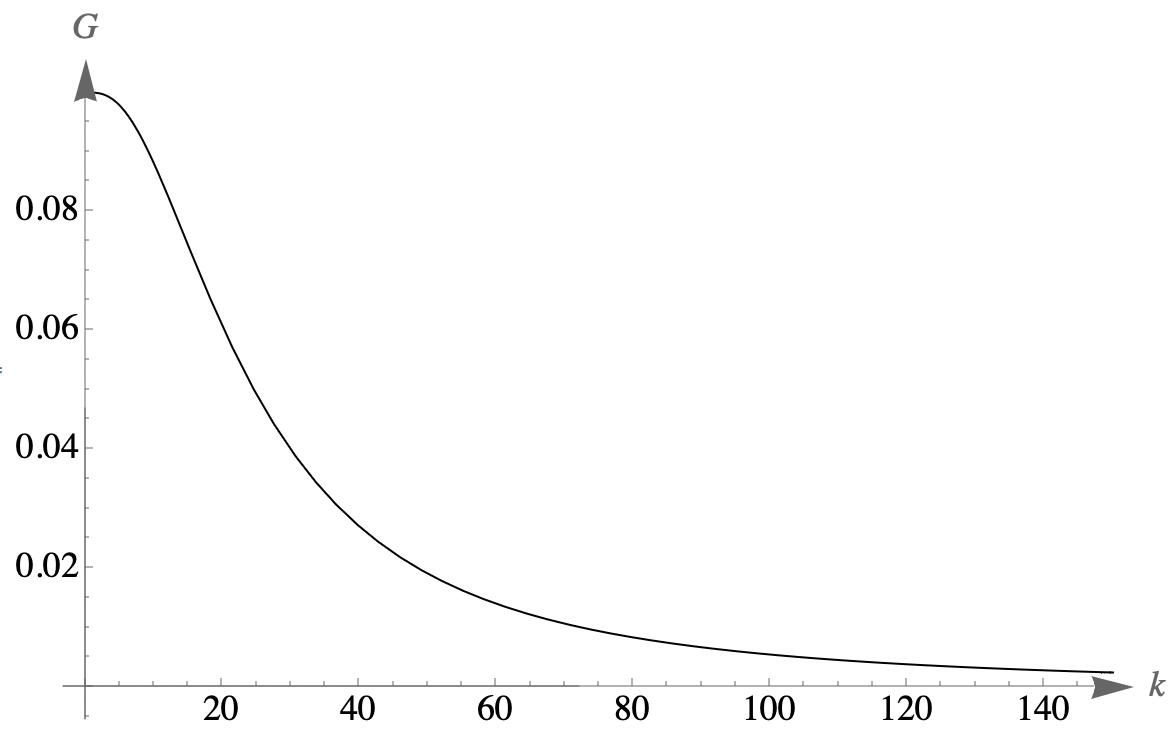}
	\caption{Plot of the $k$-dependence of the dimensionful coupling constants when $r = 0$ and $D = 3$. The trajectory chosen is the separatrix: $\Lambda_k$ vanishes in the IR and increases quadratically in the UV, while $G$ reaches a finite value in the IR and vanishes in the UV.}
	\label{fig:3}
\end{figure}

Another class of trajectories interesting to study would be the so called Type IIIa trajectories (we refer to \cite{11b} for the complete 
classification of trajectories in the Einstein-Hilbert truncation), namely those trajectories which flow towards a diverging 
$\lambda$ in the limit $k \to 0$.  In our case, these trajectories do not hit any singularity and are complete. The series in 
\eqref{b.5} has infinite radius of convergence thus we have all the means at our hand. 
Numerically, however, this turns down to be a challenging task 
because of the increasing number of orders one has to take into account: The series in (\ref{b.5}) behaves roughly as 
$\lambda^n/\sqrt{n!}$ and converges faster than geometrically for $\lambda<1$. For $\lambda\ge 1$, if we wish 
to truncate it at order $N$ and want to ensure that $\lambda^n/\sqrt{n!}\le (1/2)^n$ for $n\ge N$ 
(which estimates the error by unity) then by 
Stirling's formula we must pick $N\approx 4\; e\;\lambda^2$, i.e. $N$ grows quadratically with $\lambda$.
While it is clear that for large $k$ all trajectories approach
the fixed point, it is hard to find matching trajectories as $k\to 0$ and as $k\to\infty$.  
In Figure \ref{fig:5} and \ref{fig:6} we report the plots of the 
small $k$ regime for the dimensionful and the dimensionless coupling constants for a trajectory of the Type IIIa.

\begin{figure}[H]
	\centering
	\includegraphics[width=.48\textwidth]{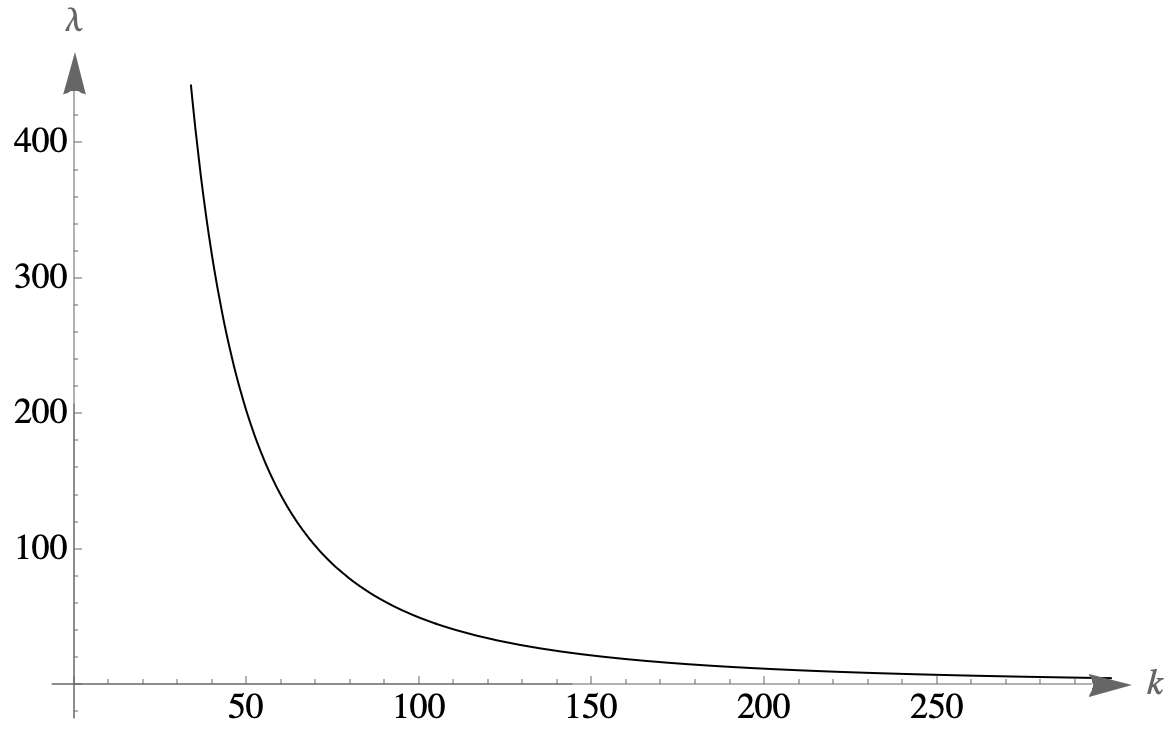}\quad
	\includegraphics[width=.48\textwidth]{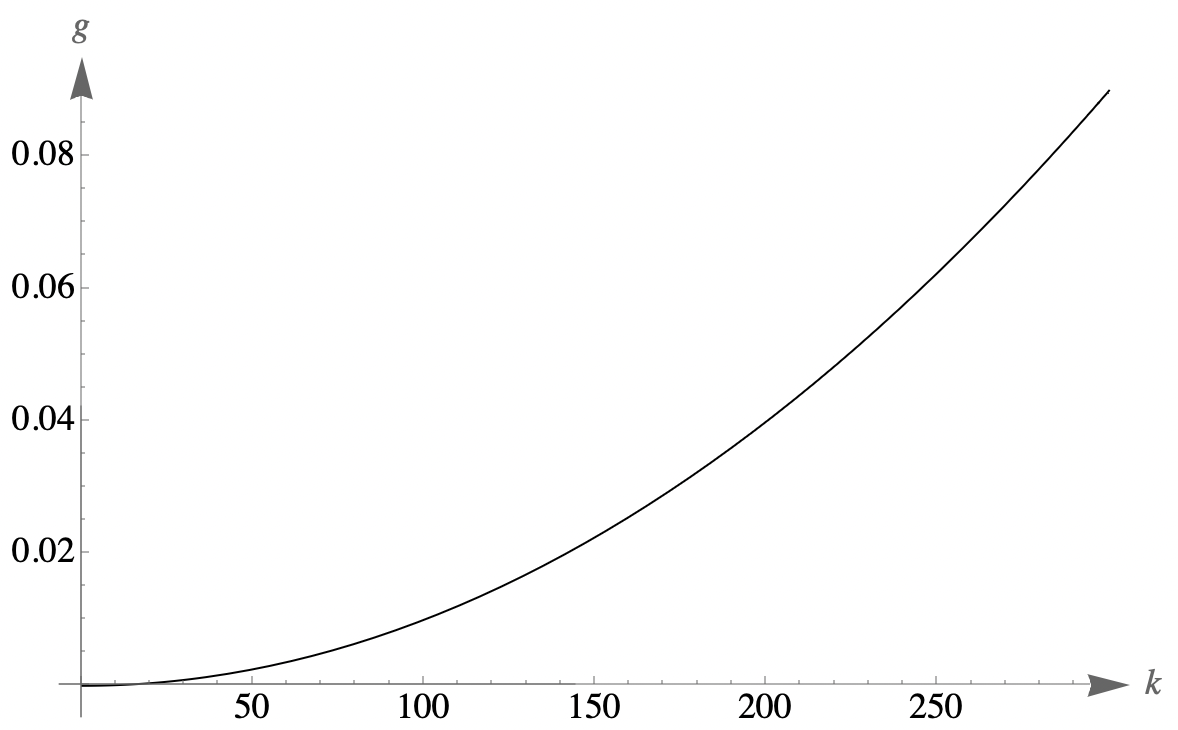}
	\caption{Small $k$ regime of the dimensionless cosmological constant and Newton's constant. The cosmological constant diverges approaching $k \to 0$, while $g$ vanishes. Note that the regime under investigation is far away from the UV fixed point. As initial conditions we picked $\lambda (k = 100) =50$ and $g(k = 100) =1/100.$}
	\label{fig:5}
\end{figure}
\begin{figure}[H]
	\centering
	\includegraphics[width=.46\textwidth]{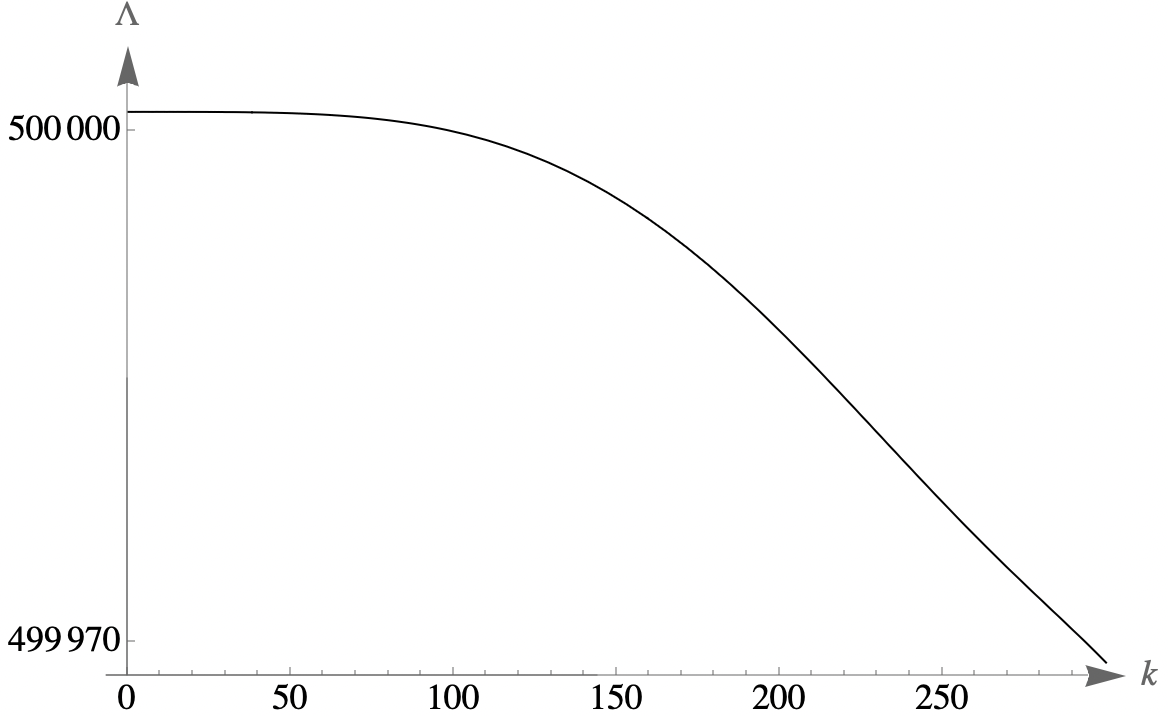}\quad
	\includegraphics[width=.49\textwidth]{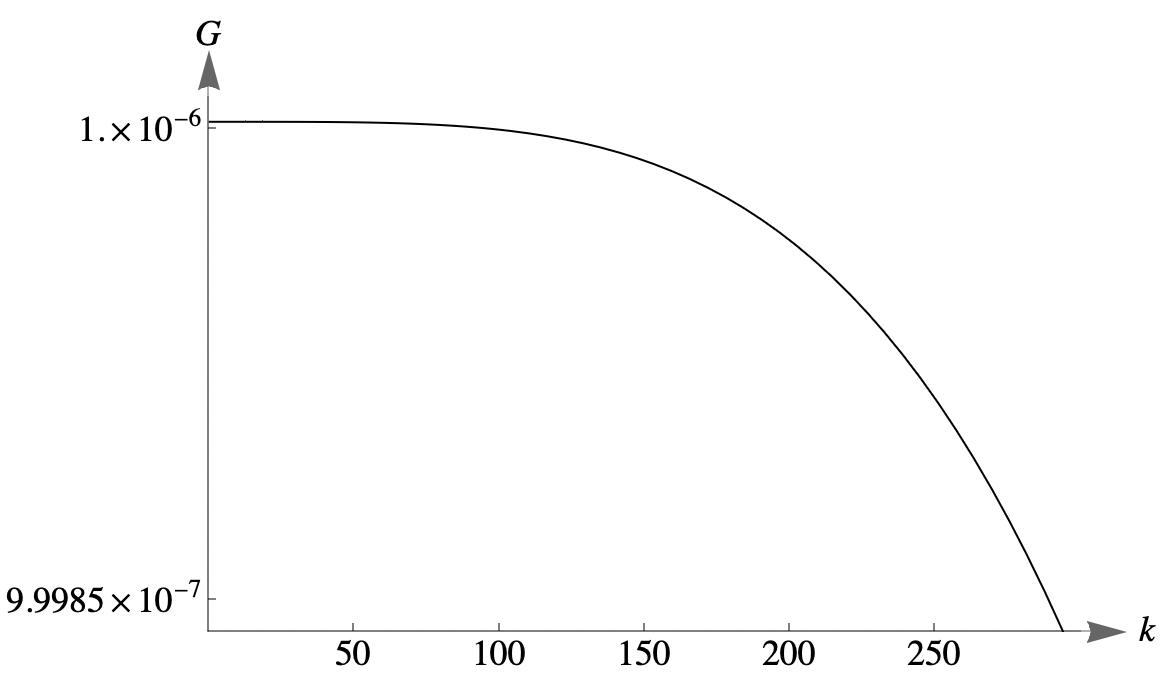}
	\caption{Small $k$ regime of the dimensionful cosmological constant and 
Newton's constant. Interestingly, both couplings reach a finite value when 
$k \to 0$. This value will depend on the initial conditions.}
	\label{fig:6}
\end{figure}

\section{Summary, conclusions and outlook}
\label{s6}

The present matter model coupled to GR comes as close as possible to the idealisation 
of a congruence of collision free observers moving on geodesics in a dynamical spacetime
while taking backreaction into account. Accordingly, its physical interpretation in terms of (Dirac) observables 
is crystal clear and there is a corresponding, distinguished 
induced physical Hamiltonian driving the dynamics of those observables. The ``problem of 
time'' is therefore solved and on the classical side one ends up a with conservative 
Hamiltonian system. Its canonical quantisation therefore does not meet any conceptual 
questions, rather there are technical issues that result from the non-polynomial 
structure of the Hamiltonian when treated non-perturbatively. In non-perturbative 
LQG one can face this 
non-polynomial structure squarely by picking a non-standard representation of the 
canonical commutation relations (CCR) and adjointness relations (AR) of the $^\ast-$algebra
of observables \cite{6}.

In non-perturbative ASQG for Hamiltonian systems one proceeds differently: In the presence of UV and IR cutoff 
one can derive a path integral formulation of the generating functional of Schwinger functions
based on the standard Schr\"odinger representation of the CCR and AR. When one formally removes both 
cutoffs one ends up with an object that is mathematically poorly defined. However, using 
average kernel techniques one can formally derive the Wetterich equation which by itself 
is well defined. One now turns the logic around and considers solutions of the Wetterich 
equation as definitions of the theory.

In the present paper we have taken the first step towards a full-fledged 
ASQG treatment of this model. New elements as compared to standard treatments 
of ASQG include:\\
i. The gauge redundancy is removed prior to quantisation. The path integral does not
contain gauge fixing and ghost terms.\\
ii. The path integral is strictly derived from the Hamiltonian formulation. That 
formulation results a priori in a phase space rather than configuration space path integral.
Luckily one can easily perform the integral over the momenta for this model which 
for generic matter coupling is very complicated, see e.g. \cite{4}.\\
iii. The resulting configuration path integral involves a non-trivial measure factor
which includes the determinant of the DeWitt metric. We can avoid that measure factor 
altogether by equipping the canonical phase space coordinates with a non-trivial density weight 
prior to quantisation. This avoids questions on how that measure factor is supposed to 
flow.\\
iv. The exponential term in the path integral is almost the Euclidean signature 
Einstein--Hilbert action but not quite: First, the Gibbons-Hawking boundary term is present and 
secondly the action is restricted to Euclidean signature metrics in synchronous gauge.\\
v. Accordingly this Euclidean action has a reduced symmetry group. When we construct 
the effective average action using a cutoff kernel only that reduced symmetry needs to 
be taken into account in order that the flow produces terms in agreement with it. 
We therefore need to construct new types of cutoffs which 
involves a natural projection operator as compared to the standard treatment. In this paper 
we have picked a choice that comes as close as possible to the standard treatment but 
involves corrections by non-minimal operators due to the unavoidable presence of the 
projection operator. The 
corrections can be dealt with using techniques already contained in the ASQG literature 
and we have confined our treatment to the leading term in this publication.\\
vi. Since we do not want to rely on unproved assumptions about the existence of 
pre-images of the Laplace transform of cutoff kernels we directly work with 
a concrete pre-image whose analytic properties are sufficient to deal with the small and 
large heat kernel time singularities of the heat kernel which requires new techniques 
in order to obtain sufficiently accurate approximants to the beta functions of the 
flow equations. 

In section \ref{s5} we performed the ASQG analysis of our model. Exploiting the possibility to perform a generic canonical 
transformation, we derive the Hessian and perform the traces via heat kernel methods. Importantly, our regulator is adapted 
to the canonical transformation in consideration. We observe that the canonical transformation breaks the foliated covariance realized 
when $r = 0$. However, we perform a covariant average in order to recast our terms in the $D+1$ Euclidean Einstein--Hilbert action, 
modulo  boundary terms, which we neglect here. We consider this averaging as part of performing the Einstein Hilbert truncation. 
When one wants to take this effect, caused by avoiding the Jacobian, properly into account and not perform 
the covariant averaging, one must expand the theory space, 
a topic that we want to examine in the future.  

Novel to our treatment is the evaluation of convoluted regularised heat kernel time integrals: we developed a technology based on the Barnes identity to expand the integrals analytically in the coupling constants to the desired order.

We specialized the study of the flow to the case $r = 0$, $D=3$ and to the case $r = \frac{D-4}{4D}$, $D=3$. 
In both cases we found the IR Gaussian fixed point and an UV attractive fixed point. 
The critical exponents are both positive and real in both cases. Comparing the  properties of the two flows, 
we do not find any qualitative large difference signaling the fact, that for this model working 
with canonical variables that avoid the otherwise present
Jacobian does not modify significantly the RG properties of the system, at least when the 
artificial covariant averaging is performed. Finally, with the purpose to access the effective 
regime at $k = 0$, we  integrate down the trajectories obtaining the values of the dimensionful 
coupling constants in this limit. Our proper time methods allow to access the effective regime for trajectories 
with positive cosmological constant along the entire flow. We analyse completely  
a trajectory of particular interest being the separatrix, connecting the UV with the IR fixed point. 
As for the trajectories of Type IIIa, due to numerical limitations, we focus the investigation the small $k$ regime. 
The large $k$ regime is dictated by the UV fixed point, from which all trajectories stem.
\\

This work can be extended and improved in various ways. An obvious 
task is to take higher order 
corrections terms in the Wetterich equation into account coming from i. geometric series expansion
of the second functional derivative,
ii. the non-minimal operator expansion mentioned above and iii. its Taylor expansion
with respect to the fluctuation
field. This requires to expand the theory space (number of symmetry consistent 
couplings) accordingly, 
moving to more general truncations. Another interesting direction would be 
to determine from a given truncation the (non-averaged $k=0$) effective action, 
to Legendre transform it to obtain the generating functional of Schwinger functions 
and to derive from that the underlying quantum Hamiltonian 
and Hilbert space representation (e.g. via Osterwalder-Schrader
reconstruction) that can then be compared to \cite{6}. We hope to come back 
to these and related questions in future publications. We believe that the present 
analysis adds to the understanding how LQG (or more generally canonical quantisation approaches 
to GR) and ASQG are related both conceptually and technically. \\  
\\ 
\\       
{\bf Acknowledgements}\\
R.F.  is supported by the FAU Emerging Talents Initiative (ETI).\\
 R.F. is grateful for the hospitality of Perimeter Institute
where part of this work was carried out.
Research at Perimeter Institute is supported in part by the Government of Canada
through the Department of Innovation, Science and Economic Development and by the
Province of Ontario through the Ministry of Colleges and Universities. This work was
supported by a grant from the Simons Foundation (grant no. 1034867, Dittrich).\\

\begin{appendix}

\section{Heat kernel of projected Laplacian}
\label{sa}

The considerations in this section are formal (i.e. without paying attention to 
functional analytic concerns). We believe that they can be made writer-tight using 
the spectral theorem applied to a self-adjoint version of the Laplacian.\\
\\
In order to construct the heat kernel of 
\be \label{a.1}
\overline{\Delta}^P_{D+1}=P\cdot \overline{\Delta}_{D+1}\cdot P, \;\;
\overline{\Delta}_{D+1}=\bar{g}^{\mu\nu} \;\bar{\nabla}_\mu\; \bar{\nabla}_\nu,\;
P^\mu_\nu=\delta^\mu_a\;\delta^a_\nu\;,
\ee
while using the known techniques to compute the heat kernel of $\overline{\Delta}_{D+1}$ we 
write 
\be \label{a.2}
e^{s\overline{\Delta}^P_{D+1}}=M(s)\; e^{s\overline{\Delta}_{D+1}},\;
M(s)=
e^{s\overline{\Delta}^P_{D+1}}\; e^{-s\overline{\Delta}_{D+1}}\;.
\ee
The M{\o}ller operator $M(s)$ obeys 
\be \label{a.3}
\frac{d}{ds} M(s)=M(s)\; A(s),\; A(s)=P(s)\cdot \overline{\Delta}^P_{D+1}\; P(s)-
\overline{\Delta}^P_{D+1},\; M(0)=1\;,
\ee
where $1$ is the identity on the type of spacetime tensor fields considered and 
\be \label{a.4}
P(s)=e^{s\overline{\Delta}_{D+1}}\; P \;e^{-\overline{\Delta}_{D+1}} 
=\sum_{n=0}^\infty\;\frac{s^n}{n!}\; [\overline{\Delta}_{D+1},P]_n
\ee
is the heat kernel evolution of the projection operator. Here 
$[\overline{\Delta}^P_{D+1},P]_0=P,\;[\overline{\Delta}^P_{D+1},P]_{n+1}=
[\overline{\Delta}^P_{D+1},[\overline{\Delta}^P_{D+1},P]_n]$ denotes the 
commutator 
of order $n$. The ODE (\ref{a.3}) has the well known solution
\be \label{a.5} 
M(s)=\sum_{n=0}^\infty M_n(s),\;\; M_0(s)=1,\; M_{n+1}(s)=\int_0^s\; dr\; M_n(r)\; A(r)\;.
\ee
The nested integrals can be computed in closed form. To that end we 
write $C_n:=[\overline{\Delta}^P_{D+1},P]_n$ and 
\be \label{a.6}
A(s)=\sum_{N=0}^\infty\; \frac{s^N}{N!} A_N,\;
A_0:=P\overline{\Delta}_{D+1} P-\overline{\Delta}_{D+1},\;
A_N=\sum_{n=0}^N \;
\left( \begin{array}{c} N\\ n \end{array} \right) 
\;
C_n\; \overline{\Delta}_{D+1}\; C_{N-n} \;\; (N>0)
\ee 
Then the first terms are
\ba \label{a.6a}
M_1(s) &=& \int_0^s\; dr\; A(r)\; =\sum_{N=1}^\infty \;\frac{s^N}{N!}\; A_{N-1} \;
\nonumber\\
M_2(s) &=& \int_0^s\; dr\; M_1(r) \; A(r)=\sum_{N=2}^\infty \frac{s^N}{N!}
\sum_{M=0}^{N-2}\; 
\left( \begin{array}{c} N-1\\ M \end{array} \right) 
\;
A_{N-2-M} \; A_M\;.
\ea   
We see that $M_n(s)$ is of order $s^n$ in heat kernel time which can be used in order to
construct the corrections to $e^{s \overline{\Delta}_{D+1}}$ which itself involves an expansion 
in terms of $s$ as described at the end of section \ref{s4}. Hence the Taylor expansions above 
and of $e^{s \overline{\Delta}_{D+1}}$  merge into a systematic expansion in $s$ in which 
the $s$ independent operator valued coefficients computed in (\ref{a.6}) act on the 
bi-tensor valued heat kernel coefficients of $e^{s \overline{\Delta}_{D+1}}$.

For the explicit computation one needs 
the $C_n$. We have e.g.
\ba \label{a.7}
C_1 &=& \bar{g}^{\mu\nu} \{\bar{\nabla}_\mu\;[\bar{\nabla}_\nu,P]+
[\bar{\nabla}_\mu,P]\;\bar{\nabla}_\nu\}=
(\overline{\Delta}_{D+1}\cdot P)+\bar{g}^{\mu\nu}\;(\bar{\nabla}_\mu\cdot P)\;\bar{\nabla}_\nu
\\
&& [\bar{\nabla}_\mu \cdot P]^t_t= [\bar{\nabla}_\mu \cdot P]^a_b= 0,\;
[\bar{\nabla}_\mu \cdot P]^t_b=-\bar{k}_{bc}\delta^c_\mu,\; 
[\bar{\nabla}_\mu \cdot P]^a_t=-\bar{k}^a_c\delta^c_\mu
\nonumber\\
&& 
[\overline{\Delta}_{D+1} \cdot P]^t_t=-2[\bar{k}^a_a]^2,\;
[\overline{\Delta}_{D+1} P]^t_a=-\bar{D}_b\bar{k}^b_a,\;
[\overline{\Delta}_{D+1} P]^a_t=-\bar{D}_b\bar{k}^{ba},\;
[\overline{\Delta}_{D+1} P]^a_b=-2\bar{k}^a_b\; \bar{k}^b_a\;,
\nonumber
\ea
where $\bar{g}_{t\mu}=\delta^t_\mu,\;\bar{g}_{ab}=\bar{q}_{ab},\; 2\bar{k}_{ab}=\bar{D}_t \bar{q}_{ab}$
was used and where $\bar{D}_a \bar{q}_{bc}=0, \bar{D}_t=\partial_t$. We see that the 
Diff$_D(M)$ correction
terms involve exactly the terms that assemble the Euclidean action. Using these 
techniques the $C_n,\;n>1$ are straightforward while tedious to compute.

\section{Barnes integral technique}
\label{sb}

In order to compute the heat kernel time $s$ integrals with respect to concrete cut-off 
functions proposed in this paper, we cannot rely on the usual methods \cite{2} that just 
work with a proposed image of the Laplace transform, assuming that 
a pre-image exists. As shown in \cite{4} the question of existence of those pre-images
is non-trivial. We therefore start from a given pre-image whose existence is thus 
secured and use this as the basis of our computation.\\
\\
The concrete heat kernel time integrals are of the type
\be \label{b.1}
J(\lambda,p):=\int_0^\infty\; ds_1\int_0^\infty\;ds_2\; e^{-s_1^2-s_1^{-2}}\; e^{-\lambda s_2}\; (s_1+s_2)^{-p}
\ee
where $\lambda>0$ and $p\in \mathbb{R}$. If $p$ is a non-positive integer then the integral 
is a sum of products of two integrals containing only $s_1,s_2$ which can be computed individually in closed
form. 
For all other cases the integral does not factorise. To factorise it for $p>0$ we make use of the 
following Barnes identity (e.g. \cite{13} and references therein).
\be \label{b.2}         
(s_1+s_2)^{-p}=\int_{-c-i\infty}^{-c+i\infty}\; dz\; B(z;s_1,s_2,p),\;\;
B(z;s_1,s_2,p)=\frac{1}{2\pi i}\; s_1^z\; s_2^{-[p+z]}\; \frac{\Gamma(z+p)\;\Gamma(-z)}{\Gamma(p)}
\ee
There are also similar identities involving 
an arbitrary number $n$ of heat kernel times which are relevant for higher order corrections 
of the Wetterich equation \cite{13}.
Here the integration path is parallel to the imaginary axis and $1>c>0$ is chosen such that $p-c$
is not a non-positive integer 
e.g. $c=\frac{1}{4}$ when $p$ is a positive 
integer or half integer. 
One can prove it using elementary Cauchy integral techniques remembering the simple pole 
structure of the $\Gamma$ function and its residua there. To do that one closes the contour via 
an infinite radius semi-circle enclosing either the positive or negative real axis using the fact 
that the $\Gamma$ function is of rapid decay at large imaginary arguments. Which 
closed path one chooses depends on $s_1,s_2$. For $s_2>s_1$ and $s_2>s_1$ respectively one closes the contour 
to the left and right respectively (for $s_1=s_2$, a set of Lebesgue measure 
$ds_1\; ds_2$ zero in $(0,\infty)^2$, we take 
a limit $s_2\to s_1 +$). In both cases one obtains a converging geometric series that can be summed 
and combines with a $p$ dependent pre-factor to the l.h.s. of (\ref{b.2}).   

We note that the roles of $s_1,s_2$ can be interchanged as the l.h.s. is invariant under this exchange.
We interchange the $z$ with the $s_1,s_2$ integrals (to be justified later) to obtain
\be \label{b.3}
J(\lambda,p)=\frac{1}{2\pi i}
\int_{-c-i\infty}^{-c+i\infty}\; dz\; \frac{\Gamma(z+p)\;\Gamma(-z)}{\Gamma(p)}\;
\left\{ \begin{array}{cc}
\lambda^{p+z-1}\; J(z)\; \Gamma(1-[p+z]) & \\ 
\lambda^{-[1+z]}\; J(-[p+z])\; \Gamma(z+1]) & 
\end{array}
\right.
\ee
where 
\be \label{b.4}
J(u)=\int_0^\infty\; ds \; e^{-s^2-s^{-2}} s^u
\ee
is analytic in $u$ and converges for any $u\in \mathbb{C}$ due to the properties of the chosen cut-off function. 
In fact, this integral is known in terms of modified Bessel functions \cite{28}.

Both ways of writing the integral (\ref{b.3}) are a priori equally valid for any value of $\lambda>0$.
However when carrying out the $z$ integral via the residue theorem we obtain a 
series which whose rate of convergence depends on whether $\lambda<1$ or $\lambda>1$.
Specifically, closing the contour to the right or left respectively yields pole contributions from 
positive and negative integers respectively, hence to improve convergence one may choose to close 
the contour to the right and left respectively for $\lambda<1$ and $\lambda>1$ respectively when 
using the upper version of (\ref{b.3}) while one may choose to close 
the contour to the left and right respectively for $\lambda<1$ and $\lambda>1$ respectively when 
using the lower version of (\ref{b.3}).

It remains to investigate the pole structure of the $z$ integral and to close the contour.  
We confine ourselves to the upper version of (\ref{b.3}) and are thus 
confronted with poles of $\Gamma(z+p)\Gamma(-z)\Gamma(1-p-z)$.
$\lambda<1$:\\ 
As indicated, we close the contour to the right where $\Gamma(z+p)$ is entire. 
For $p=1$ the poles are at $z=0,1,,2,..$ and are double.
For an integer $p>1$ the poles at $z=0,1,..,p-2$ are single and those for $z=p-1,p, ..$ 
are double. For $p>0$ not an integer the poles at $z=0,1,2,..$ and $z=p-1,p,..$ are 
both single valued. \\
$\lambda>1$:\\ 
We close the contour to the left. Then $\Gamma(-z)$ is entire and $\Gamma(1-p-z)$ has 
poles at $z=-1,-2,.., -(p-1)$ while $\Gamma(z+p)$ has poles at $z=-p,-p-1,..$. Thus 
all poles are simple in this case. \\ 
The existence of double poles for the right contour leads to derivatives with respect to the holomorphic
part of the integrand in (\ref{b.3}) when applying the residue theorem and thus to terms 
proportional to $\ln(\lambda)$. In either case one obtains a series 
in $\lambda$ or $\lambda^{-1}$ respectively whose coefficients can be computed in closed 
form. For instance for $p>0$ an integer and $\lambda<1$ 
\ba \label{b.5}
J(\lambda,p) &=&        
\frac{(-1)^p}{(p-1)!}\;\sum_{n=0}^\infty\; \frac{\lambda^{n+p-1}}{n!}\;J(n)\;
\{\frac{J'(n)}{J(n)}+\ln(\lambda)+\frac{c_{n+p-1}}{(n+p-1)!}
-[2 c_0+\sum_{k=0}^{n-1}\;\frac{1}{n-k}+\sum_{k=0}^{n+p-2}\;\frac{1}{n+p-1-k}]\}
\nonumber\\
c_n &:=& \int_0^\infty\; ds\; s^n \; \ln(s)\; e^{-s}
\ea
while for $\lambda>1$
\be \label{b.6}
J(\lambda,p) = \frac{(-1)^p}{(p-1)!}\; \sum_{n=1}^\infty\; (-1)^n\;\lambda^{p-1-n}\; \Gamma(n)\;
J(-n)
\ee
It is understood that the sums over $k$ displayed in (\ref{b.5}) are missing when the upper 
summation bound is lower than zero and $J'$ is the derivative of $J$. By combining formulae 
(10.32.9) and (10.41.2) of \cite{29} (see also \cite{28}) one finds that $J(-n)=J(n-2)$ and
$J(n)=K_{(n+1)/2}(2)$ where $K_\nu(z)$ is the modified Bessel function of the second 
kind which has an asymptotics for large $\nu$ and fixed $z$ such that
$J(n)\sim \Gamma((n+1)/2)/\sqrt{n+1}$. It is not difficult to see that the curly bracket 
in (\ref{b.5}) grows at most linearly with $n$. Thus (\ref{b.5}) has infinite radius of absolute convergence
while (\ref{b.6}) has zero radius of absolute convergence and interchange of the integrals 
is justified only for (\ref{b.5}). We interpret (\ref{b.6}) as an asymptotic series presentation 
of the function (\ref{b.5}).   

When solving the flow equation up to the first order in truncation, we are also facing the challenge of computing 
integrals containing the convolution of three and four heat kernel times (see \eqref{exp}). 
However, one can show that the following integral over three convoluted heat kernel times can be reduced to 
an integral over two convoluted heat kernel times, by introducing $s_4=s_2+s_3$ instead of $s_3$ as an integration 
variable, noticing that the integrand only depends on $s_4$ and that $s_2$ is confined by $s_2\le s_4$ hence  
integrating by parts:
\begin{equation}\label{b7}
\int_0^\infty ds_1\int_0^\infty ds_2\int_0^\infty ds_3  e^{-s_1^2-s_1^{-2}}e^{-\lambda (s_2+s_3)}(s_1+s_2+s_3)^{-p} = \int_0^\infty ds_1\int_0^\infty ds_2 e^{-s_1^2-s_1^{-2}}s_2e^{-\lambda (s_2+s_3)}(s_1+s_2)^{-p} 
\end{equation}
after renaming $s_4\to s_2$.
Hence, these terms with originally three heat kernel integrations can be computed with the methods above. 
The contribution with originally four heat kernel times contains two insertions of the regulator and two insertions of the propagator, 
which can be reduced to insertion of one propagator using \eqref{b7}. However, the two insertions of the regulator require a little more work.
Structurally, we have
\begin{equation}
\int_0^\infty ds_1\int_0^\infty ds_2 e^{-s_1^2-s_1^{-2}}e^{-s_2^2-s_2^{-2}}e^{-\lambda s_3}(s_1+s_2+s_3)^{-p} 
\end{equation}
Again we apply the Barnes identity \eqref{b.2} on $s_1+s_2$ together as a single variable and $s_3$. After having chosen the appropriate contour and evaluated the residue at the respective poles, this leads then to a similar structure as \eqref{b.5}, with the difference that there is one more class of coefficients to evaluate, namely
\begin{equation}
d_n := \int_0^\infty ds_1 \int_0^\infty ds_2\; e^{-s_1^2-s_1^{-2}}e^{-s_2^2-s_2^{-2}} \ln (s_1+s_2) \;(s_1+s_2)^n\;.
\end{equation}
Hence, by means of the Barnes identity, we have developed a method to evaluate convoluted heat time integrals at every order needed in this paper.

\end{appendix}


\end{document}